\documentclass{article}

\pdfoutput=1
\usepackage{arxiv}

\usepackage[utf8]{inputenc} 
\usepackage[T1]{fontenc}    
\usepackage{hyperref}       
\hypersetup{pdfauthor={Tony Ware},pdfstartview=FitH,pdfpagemode=none,bookmarksopen=false,colorlinks=true,urlcolor=blue,linkcolor=blue,citecolor=blue}

\usepackage{url}            
\usepackage{booktabs}       
\usepackage{amsfonts}       
\usepackage{nicefrac}       
\usepackage{microtype}      
\usepackage{lipsum}		
\usepackage{graphicx}
\usepackage{doi}
\usepackage[section]{placeins}
\usepackage[utf8]{inputenc}
\usepackage{amsmath,amsthm}
\usepackage{array}
\usepackage{wrapfig}
\usepackage{multirow}
\usepackage{tabularx}
\usepackage[nottoc,numbib]{tocbibind}
\usepackage{amssymb}
\usepackage{amsmath}
\usepackage{amsmath,amssymb,amsthm,mathrsfs,dsfont,epsfig,subfigure,fullpage}
\usepackage{graphics}
\usepackage{graphicx}
\usepackage{amsfonts}
\usepackage{fancyhdr}
\usepackage{latexsym,amssymb,amsmath,amsopn,bm}
\usepackage{amsthm}
\usepackage{siunitx}
\usepackage{comment}
\usepackage{mathrsfs,dsfont,epsfig,subfigure,fullpage}
\usepackage{fancyhdr}
\usepackage{fixmath}
\usepackage{color, bm, comment}
\usepackage{commath}
\usepackage[toc,page]{appendix}

\usepackage[utf8]{inputenc}
\usepackage[english]{babel}
\usepackage{siunitx}
\newtheorem{theorem}{Theorem}[section]
\newtheorem{lemma}[theorem]{Lemma}
\newtheorem{proposition}[theorem]{Proposition}
\newtheorem{corollary}[theorem]{Corollary}

\DeclareMathOperator{\sgn}{sgn}
\newcommand{\word}[1]{\;\;\text{#1}\;\;}

\newcommand{\ps}[2]{{_{_{#1}}\!\!}#2}

\def\E{\mathbb{{E}}}
\def\V{\mathbb{{V}}}
\def\N{\mathbb{{N}}}
\def\Cov{\text{Cov}}

\theoremstyle{}
\newcommand{\thistheoremname}{}
  \newtheorem*{genericthm*}{\thistheoremname}
\newenvironment{namedthm*}[1]
  {\renewcommand{\thistheoremname}{#1}%
   \begin{genericthm*}}
  {\end{genericthm*}}

\title{A Weighted Multilevel Monte Carlo Method}

\author{
Yu Li \\
	Department of Mathematics and Statistics\\
	University of Calgary\\
	\texttt{yu.li1@ucalgary.ca} \\
\And
	Antony  Ware\thanks{Dept.~of Mathematics and Statistics, University of Calgary, 2500 University Drive NW, Calgary, Alberta , Canada. T2N 1N4} \\
	Department of Mathematics and Statistics\\
    University of Calgary\\
	\texttt{aware@ucalgary.ca} \\ 
}

\hypersetup{
pdftitle={A Weighted Multilevel Monte Carlo Method},
pdfsubject={},
pdfauthor={Antony Ware, Yu Li},
pdfkeywords={Monte Carlo, Multilevel Monte Carlo},
}

\begin{document}
\maketitle

\begin{abstract}
The Multilevel Monte Carlo (MLMC) method has been applied successfully in a wide range of settings since its first introduction by Giles \cite{Giles2008}. When using only two levels, the method can be viewed as a kind of control-variate approach to reduce variance, as earlier proposed by Kebaier \cite{Kebaier2005}. We introduce a generalization of the MLMC formulation by extending this control variate approach to any number of levels and deriving a recursive formula for computing the weights associated with the control variates and the optimal numbers of samples at the various levels. 

We also show how the generalisation can also be applied to the \emph{multi-index} MLMC method \cite{Haji-AliNobileTempone2015}, at the cost of solving a $(2^d-1)$-dimensional minimisation problem at each node when $d$ index dimensions are used.

The comparative performance of the weighted MLMC method is illustrated in a range of numerical settings. While the addition of weights does not change the \emph{asymptotic} complexity of the method, the results show that significant efficiency improvements over the standard MLMC formulation are possible, particularly when the coarse level approximations are poorly correlated.
\end{abstract}

\keywords{Monte Carlo \and Multilevel Monte Carlo \and Control Variates \and Stochastic Differential Equations}

\section{Introduction}

The Multilevel Monte Carlo method targets a situation where the goal is to compute an expectation $\overline{P} = \E[P]$ via Monte Carlo estimation, but where $P$ is either impossible or prohibitively expensive to sample directly. Instead, one has a sequence of approximations $P_l$ that may be sampled with a computational cost that increases with the index $l$, with $\overline{P}_l=\E[P_l]$ converging to $\overline{P}$ as $l$ increases. 

With single-level Monte Carlo estimate, the total error is controlled by choosing the level $l$ to ensure an acceptable bias, and increasing the number of samples to reduce the variance of the estimate. A typical scenario leads to a computational cost of $O(\epsilon^{-3})$ in order to achieve an overall error of $\epsilon$. The MLMC method offers a way to improve on this by combining the estimate at the finest level with differences of estimates at coarser levels, generating more samples at the coarser, cheaper, levels, and fewer at finer, more expensive levels. The result is a method with the bias of the finest level, but potentially much lower variance for a given computation effort, and an computational cost that can be as low as $O(\epsilon^{-2})$, matching the best-possible accuracy for Monte Carlo estimation (see, for instance, Theorem~2.1 of \cite{Giles2015}). 

The potential of the MLMC approach has lead to hundreds of papers dealing with refinements and generalisations of the original approach, along with a host of application areas. Giles \cite{Giles2015} provided a comprehensive overview of the development of MLMC methods in the first few years following its first appearance in \cite{Giles2008}. Generalisations included randomised MLMC \cite{RheeGlynn2012} (for truly unbiased estimates), Multilevel quasi-Monte Carlo (qMLMC) \cite{GilesWaterhouse2009}, combinations of MLMC with Richardson extrapolation \cite{LemairePages2017}, and Multi-index MLMC \cite{Haji-AliNobileTempone2015} (for situations where there is more than one `scale' that can be adjusted in the approximate samples, such as might arise in simulations of stochastic partial differential equations \cite{ReisingerGiles2012,BarthLang2012,ReisingerWang2018,LangPetersson2018,Petersson2020,ChadaHoelJasra2022}). 

Other applications of MLMC and qMLMC include their use for estimation, uncertainty quantification, optimal control in the context of PDEs with random coefficients \cite{CliffeGilesScheichl2011,BarthSchwabZollinger2011,AbdulleBarthSchwab2013,TeckentrupScheichlGiles2013,KuoSchwabSloan2015,AliUllmannHinze2017,LuoWang2019,Van-BarelVandewalle2019}. They have also been applied in settings involving nested expectations\cite{Giles2018}, including Picard iterations for solving semilinar parabolic PDEs \cite{WeinanHutzenthalerJentzen2021}.
In recent years there have been some applications of MLMC method to Machine Learning and Deep Learning techniques, for instance \cite{Jabarullah-KhanElsheikh2019} on uncertainty quantification, \cite{WangWang2023} for intractable distributions and \cite{FujisawaSato2021} for variational inference frameworks.

In order to fix ideas here, we will focus on the applications of MLMC described in Giles' original papers \cite{Giles2008,Giles2008a}, i.e.~to option valuation for assets governed by one-dimensional stochastic differential equations (SDEs).
To that end, suppose that an asset price $S_t$ is governed by a one-dimensional It\^o diffusion 
\begin{equation}
    dS_t=a(S_t)dt+b(S_t)dW_t, \hspace{10pt} 0 \leq t \leq T,
    \label{eq1}
\end{equation}
where $W_t$ is a standard Brownian motion under the pricing measure, $a$ and $b$ are given functions satisfying certain conditions, and the initial value $S_0$ is assumed to be given. We may wish to compute the value of a European (possibly path-dependent option with payoff $\Lambda(S_{[0,T]})$, where $S_{[0,T]}$ denotes the path followed by $S_t$ over the interval $[0,T]$, given a realisation of the Brownian motion $W_t$ on that interval. For example, a call option with strike price $K$ will have $\Lambda(S_{[0,T]}) = \max(S_T-K,0)$, while for an Asian put option we would have $\Lambda(S_{[0,T]}) = \max\left(\frac{1}{T}\int_0^TS_t dt -K,0\right)$.
The option value will be given by a discounted expected payoff (we assume a constant interest rate $r$), so that in this setting $P = e^{-rT}\Lambda(S_{[0,T]})$.

Approximate samples $P_l$ ($l=0,\dots,L$) for $P$ might arise from simple Euler-Maruyama discretisations of \eqref{eq1}, with $P_l$ using $J_l$ steps and a constant time step $h_l=T/J_l$ (other discretisations, such as the Milstein discretisation, c.f.~\cite{Giles2008a}, can also be used). We define $t^l_j = jh_l, \; j=0,\ldots,J_l$ and, given the initial value $S^l_0=S_0$, define the values of the discrete path $S^l_{[0,J_l]} = (S^l_0,\ldots,S^l_{J_l})$ via
\begin{equation}
\label{eq:EM}
    S^l_{j+1}=S^l_j+a(S^l_j)h+b(S^l_j)\Delta W^l_j, \hspace{10pt} j=0,1...,J_l-1.
\end{equation}
where $\Delta W^l_j=W_{t^l_{j+1}}-W_{t^l_j}$. For such a sequence, we can define a discrete payoff function $\Lambda^l$ that should correspond in an appropriate way to $\Lambda$. (In the case of a call option, $\Lambda^l(\widehat{S}_{[0,J]}) = \max(S^l_J-K,0)$.) Then, if we have a sequence of grids with $J_l$ steps, $l=0,1,\dots$, we can define $P_l = e^{-rT}\Lambda^l(S^l_{[0,J_l]})$.

A single-level Monte Carlo estimate for $Y_l$, using $N_l$ independent sample paths, will be formed as the average of the values of $P_l$ resulting from those paths, and we denote such an average 
$\displaystyle{\ps{N_l}{P_l}}$.

Then $\E[\ps{N_l}{P_l}] = \overline{P}_l$, and the variance of $\ps{N_l}{P_l}$ is
\begin{equation*}
    \V[\ps{N_l}{P_l}]=\frac{1}{N_l}\mathbb{V}[P_l].
\end{equation*}
The mean square error (MSE) associated with our approximation is then
\begin{equation}
\label{eq:MSE}
  \E\left[\left({_{N_l}}P_l-\overline{P}\right)^2\right] =\frac{1}{N_l}\V[P_l]+(\overline{P}_l-\overline{P})^2.
\end{equation}
It is well established (c.f.~\cite{BallyTalay1995}, for instance) that, provided $a$, $b$ and the payoff function $\Lambda$ satisfy certain conditions, the Euler-Maruyama discretisation of \eqref{eq1} will have a weak order of convergence of 1, so that $\E[P-P_l]^2=O(J_l^{-2})$. Given that $\V[{_{N_l}}P_l] = O(N_l^{-1})$, the MSE will be $O(N_l^{-1})+O(J_l^{-2})$.
The total computational cost will be proportional to $N_l J_l$, and it is straightforward to show that, to achieve an MSE of $\epsilon^2$ while minimising the computational effort, we should have $N_l$ proportional to $\epsilon^{-2}$ and $J_l$ proportional to $\epsilon^{-1}$, giving a total computational cost proportional to $\epsilon^{-3}$.

Some early steps  toward reducing this computational complexity involved the use of approximations at two levels, with the coarser level approximation serving as a control variate. This was the approach taken by Kebaier \cite{Kebaier2005}, who proved that the the total cost can be reduced to $O(\epsilon^{-2.5})$ through an appropriate combination of two estimators with time steps $h$ and $\sqrt{h}$. His work was an application of a more generally-applicable approach of quasi-control variates analysed by Emsermann and Simon \cite{Emsermann2002}. Giles' Multilevel Monte Carlo approach can in some cases reduce this cost to $O(\epsilon^{-2})$.

\subsection{Multilevel Monte Carlo}
The multilevel approach proposed by Giles \cite{Giles2008} involves forming a geometric sequence of grids using $J_l = J_0 M^l$ time steps, for some integer $M$ (typically $M=2$ or $M=4$), and for $l=0,1,2,\dots,L$. These grids are then used to construct \emph{basic estimators} $Y_0 = P_0$, and, for $l>0$, $Y_l = P_l - P^l_{l-1}$, where $P^l_{l-1}$ is an estimate constructed on a grid with timestep $h_{l-1}$, but using the \emph{same random samples} $\Delta W^l_j$ as were used to compute $P_l$. This can be done, for instance, by summing successive samples, so that, if $M=2$, we can set $\Delta W^{l-1}_j = \Delta W^l_{2j}+\Delta W^l_{2j+1}$, for $j = 0,\dots,J_{l-1}-1$ and use these to compute a discrete path with timestep $h_{l-1}$, which can then be used to compute $P^l_{l-1}$.

In this way, we find that $\E[Y_l] = \overline{P}_l - \overline{P}_{l-1}$, for $l>0$, and this means that a telescoping sum of Monte Carlo samples of the basic estimators, using $N^L_l$ samples for $Y_l$, i.e. 
\begin{equation}
\label{eq:MLMCP} 
{\cal P}_L = \ps{N^L_0}{Y_0} + \ps{N^L_1}{Y_1} + \dots + \ps{N^L_L}{Y_L}
\end{equation}
satisfies $\E[{\cal P}_L] = \overline{P}_L$. 

If we denote the computational effort needed to compute a single sample of $Y_l$ by $\eta_l^2$, and the standard deviation of $Y_l$ by $\Delta_l$, then the computational effort required to compute  ${\cal P}_L$ is 
\[ E_L^2 = \sum_{l=0}^LN^L_l\eta_l^2\]
and the variance of ${\cal P}_L$ is given by
\[ \V[{\cal P}_L] = \sum_{l=0}^L\frac{\Delta_l^2}{N^L_l}.\]
The minimal (square root) cost $E_L$ needed to achieve a variance of $v^2$ is
\begin{equation}
\label{eq:MLMCE}
E_L = \frac{1}{v}\sum_{l=0}^L\Delta_l\eta_l,
\end{equation}
and this is achieved by setting 
\begin{equation}
\label{eq:MLMCN} N^L_l = \frac{E_L\Delta_l}{\eta_lv}, \quad l=0,1,\dots,L.
\end{equation}
(This can be readily derived using a Lagrange multiplier approach.)
\subsubsection{Complexity Theorem}
The potential savings of the approach, and the circumstances under which they may be realised, are shown in the following theorem (c.f.~Theorem~2.1 of \cite{Giles2015}).
\begin{theorem}
\label{thm:MLMCComplexity}
Suppose that there are positive constants $c_1, c_2, c_3, \alpha,\beta,\gamma$ with $2\alpha\geq\min(\beta,\gamma)$ such that, for $l=0,1,2,\dots$,
\begin{enumerate}
\item $\abs{\overline{P}_l-\overline{P}}\leq c_1 2^{-\alpha l} $,
\item $\Delta^2_l \leq c_2 2^{-\beta l}$, and
\item $\eta^2_l \leq c_3 2^{\gamma l}$.
\end{enumerate}
Then there exists a positive constant $c_4$ such that, for any $\epsilon<1/e$ there exists an integer $L$ such that ${\cal P}_L$ (as defined by \eqref{eq:MLMCP}--\eqref{eq:MLMCN}) has mean squared error with bound
\[ \E[({\cal P}_L-\overline{P})^2] < \epsilon^2\]
with a computational complexity $C$ with bound
\[ C \leq \begin{cases}
c_4\epsilon^{-2} & \beta>\gamma,\\
c_4\epsilon^{-2}(\log\epsilon)^2 & \beta=\gamma,\\
c_4\epsilon^{-2-\frac{\gamma-\beta}{\alpha}} & \beta<\gamma.
\end{cases}
\]
\end{theorem}
The conditions of the theorem constrain the rate of decay of the bias of $P_l$ and the standard deviation of the basic estimators $Y_l$ and the rate of growth in the cost of sampling the estimators. In the `ideal' case, where $\gamma<\beta\leq 2\alpha$, we see that, for a given $L$,
\[ E_L \leq \frac{c_2c_3}{v}\sum_{l=0}^L2^{-(\beta-\gamma)l/2} \leq \frac{c_2c_3}{v}\frac{1}{1-2^{-(\beta-\gamma)/2}}.\]
Thus the cost is proportional to $1/v^2$ for any $L$. We can thus set\footnote{Note that it is possible to optimise over the proportion of $\epsilon^2$ that is split between the error sources.} $v=\epsilon/\sqrt{2}$ and choose $L$ so that $\abs{\overline{P}_l-\overline{P}}\leq\epsilon/\sqrt{2}$ to satisfy the conditions of the theorem and verify its conclusions in this case. 

\subsubsection{Performance at the coarsest levels}
While the asymptotic complexity of the MLMC method is an improvement over that of single-level Monte Carlo (in almost all cases), in practice, it may be optimal to reduce the number of coarsest levels used.  In particular, unless the correlation between $P_1$ and $P^1_0$ is sufficiently high, it will be better to set level $l=1$ to be the coarsest level. 

If we denote the correlation between $P_l$ and $P^l_{l-1}$ by $\rho_l$, and the standard deviation of $P_l$ by $\sigma_l$, then we have, with $l=0$ the coarsest level, 
\[ \Delta_0 = \sigma_0 \word{and} \Delta^2_1 = \sigma_1^2-2\rho_1\sigma_1\sigma_0+\sigma_0^2.\]
If $l=1$ is the coarsest level, then $\Delta_1 = \sigma_1$. The total effort for an MLMC estimate starting from level $l=1$ will be lower than that for an estimate starting from $l=0$ unless
\[ \eta_0\sigma_0+ \eta_1\sqrt{\sigma_1^2-2\rho_1\sigma_1\sigma_0+\sigma_0^2}\leq \eta_1\sigma_1.\]
Moving the $\eta_0\sigma_0$ term to the right hand side, and squaring, we find that the condition becomes
\[ \rho_1\geq \frac{\sigma_0}{\sigma_1}\frac{2^{\gamma}-1}{2^{\gamma+1}} + \frac{1}{2^{\frac{\gamma}{2}}},\]
where we have assumed a) that the cost of computing a sample of $Y_1=P_1-P^1_0$ is similar to the cost of computing $P_1$, and b) that $\eta_1/\eta_0 = 2^{\frac{\gamma}{2}}$. 

As an example, if $\sigma_0=\sigma_1$ and $\gamma = 1$, then the correlation needs to be above roughly $0.957$ in order for the level $0$ samples to add any benefit. Moreover, if in fact  a sample of $Y_1=P_1-P^1_0$ is more expensive to compute than a sample of $P_1$, this lower bound will be higher.

\section{Weighted Multilevel Monte Carlo (WMLMC)}
Here we propose a generalisation of the standard MLMC method that involves rewriting the MLMC estimators as nested control variate estimators, and (in that spirit) adding weights. The impact of this is to enable more efficient use to be made of the coarsest level samples, especially when the correlation between the levels is relatively low,  with potentially significant cost savings (although with no impact on the asymptotic order of complexity).

Recent work by Amri et~al.~\cite{AmriMycekRicci2023} also takes a multilevel control variate point of view, although their approach differs from the method proposed here by using extrinsic surrogate estimators as control variates.

A weighted version of MLMC was also proposed in \cite{LemairePages2017}, where the possibility of adding weights in order to better control the bias was introduced. There a \emph{design matrix} ${\bf T}$ was used to define a general multilevel estimator. Several templates for ${\bf T}$ were discussed (c.f.~Section~3.32 and Section~5.1). However, a defining property of these design matrices was that their columns (apart from the first) should all sum to zero. The weighting scheme we propose here breaks that property; its goal is to improve the variance, rather than the bias.

\subsection{A recursive formulation of MLMC}
We start by showing that the MLMC estimators can be written recursively, which is possible because the ratio between successive numbers of samples $N^L_l$ is independent of the finest level $L$. We define, for $l=0,\dots,L$, $E_l=\frac{1}{v}\sum_{l'=0}^l\Delta_{l'}\eta_{l'}$ (c.f.~\eqref{eq:MLMCE}), and set 
\begin{equation}
\label{eq:MLMCalphabeta}
\alpha_l:= \frac{E_l\Delta_l}{\eta_l v}, \; l=0,\dots,L \word{and}\beta_l:=  
\frac{E_l}{E_{l-1}},\; l=1,\dots,L.
\end{equation}
(Note that $E_0 = \sigma_0\eta_0/v$, so that $\alpha_0 = \sigma_0^2/v^2$.)
Then we can recursively define the multilevel estimators
\begin{equation}
\label{eq:MLMCRecursive}
\begin{aligned}
{\cal P}_0 &= \ps{\alpha_0}{Y_0}, \\
{\cal P}_l &= \ps{\alpha_l}{Y_l}+\ps{\beta_l\,}{{\cal P}_{l-1}},\quad l=1,\dots,L.
\end{aligned}
\end{equation}
The fact that we use the same notation for these multilevel estimators as for the MLMC estimators defined in \eqref{eq:MLMCP}--\eqref{eq:MLMCN} is justified by the following lemma.
\begin{lemma}
Let $L\in\N_0$. The multilevel estimators defined in \eqref{eq:MLMCRecursive}, with $\alpha_l$ and $\beta_l$ given by \eqref{eq:MLMCalphabeta}, satisfy, for $l=0,\dots,L$,
\begin{align}
\label{eq:MLMClem1} \V[ {\cal P}_l ] &= v^2\\
\label{eq:MLMClem2} \E[{\cal P}_l] &= \overline{P}_l\\
\label{eq:MLMClem3}
{\cal P}_l &= \sum_{l'=0}^l \ps{N^l_{l'}\,}{Y_{l'}},
\end{align}
where $N^l_{l'} = \alpha_{l'}\beta_{l'+1}\dots\beta_{l}$. 
In particular, ${\cal P}_L$ satisfies \eqref{eq:MLMCP}, with $N^L_l$ given by \eqref{eq:MLMCN}.
\end{lemma}
\begin{proof}
The equality \eqref{eq:MLMClem3} follows from expanding \eqref{eq:MLMCRecursive}. Then \eqref{eq:MLMClem2} follows via the telescoping sum of means of $Y_l$. Noting that
\[ N^l_{l'} = \alpha_{l'}\beta_{l'+1}\dots\beta_{l} = \frac{E_{l'}\Delta_{l'}}{\eta_{l'}v}\frac{E_{l'+1}}{E_{l'}}\dots\frac{E_l}{E_{l-1}} = \frac{E_l\Delta_{l'}}{\eta_{l'}v},
\]
we see that \eqref{eq:MLMClem2} follows. At this point, the fact that ${\cal P}_L$ as defined via \eqref{eq:MLMCRecursive} satisfies \eqref{eq:MLMCP} follows immediately.
\end{proof}
\subsubsection{MLMC and control variates}
The final step of the recursion \eqref{eq:MLMCRecursive} can be expanded to give
\begin{equation}
\label{eq:MCMCCV}
{\cal P}_L = \ps{\alpha_L}{Y_L}+\ps{\beta_L}{{\cal P}_{L-1}} = \ps{\alpha_L}{P_L} - \left(\ps{\alpha_L}{P^L_{L-1}}-\ps{\beta_L}{{\cal P}_{L-1}}\right).
\end{equation}
The term in brackets has mean zero and can be seen to act as a control variate, correlated with the finest level single-level estimator $\ps{\alpha_L}{P_L}$. This can be continued at each level, so that the MLMC estimator can be seen as a nested sequence of control variates, each employed to reduce the variance of the finer level estimator.

At this point, we are naturally lead to consider adding weights to improve the efficiency of these control variates.
\subsection{Adding weights}
Given a set of weights $\theta = (\theta_0,\dots,\theta_L)$, with $\theta_0=0$, we define the estimators ${\cal P}^\theta_l$ for our \emph{weighted} multilevel Monte Carlo method (WMLMC) recursively.
\begin{equation}
\label{eq:MLMCW}
\begin{aligned}
{\cal P}^\theta_0 &= \ps{\alpha_0}{Y^\theta_0}=\ps{\alpha_0}{P_0}, \\
{\cal P}^\theta_l &= \ps{\alpha_l}{P_l} - \theta_l\left(\ps{\alpha_l}{P^l_{l-1}}-\ps{\beta_l}{{\cal P}^\theta_{l-1}}\right)= \ps{\alpha_l}{Y^\theta_l}+\theta_l\,\ps{\beta_l\,}{{\cal P}^\theta_{l-1}},\quad l=1,\dots,L,
\end{aligned}
\end{equation}
so that, here $Y^\theta_l = P_l-\theta_lP^l_{l-1}$ for $l>0$, and the values of $\alpha_l$ and $\beta_l$, along with those of $\theta_l$, will be determined so as to minimise the cost (for a given variance of $v^2$) of computing each ${\cal P}^\theta_l$. 

We see that we can write 
\begin{equation}
\label{eq:MLMCWexp}
{\cal P}^\theta_l = \sum_{l'=0}^l\Theta^l_{l'}\; \ps{N^l_{l'}\,}{Y^\theta_{l'}},\quad l=0,\dots,L
\end{equation}
with, as before $N^l_{l'} = \alpha_{l'}\beta_{l'+1}\dots\beta_{l}$, $0\leq l'\leq l\leq L$ and
\begin{equation}
\label{eq:MLMCWTheta}
\Theta^l_l = 1, 0\leq l\leq L, \word{and} 
\Theta^l_{l'} = \prod_{k=l'+1}^l\theta_k,\quad 0\leq l'< l\leq L.
\end{equation}
If we try to optimise the values of $\Theta^L_{l}$ and $N^L_l$ directly, we find that they satisfy a coupled system of nonlinear equations. It is possible (\cite{WareLi2021}) to solve these equations using a sweeping approach reminiscent of the Thomas algorithm for tridiagonal systems. However, it is more natural to use the recursive equations \eqref{eq:MLMCW}, and the resulting optimal values are provided in the Proposition~\ref{prop:WMLMC} below.

We denote the cost of generating a sample of ${\cal P}^\theta_l$ by $E^\theta_l$, and the standard deviation of $Y^\theta_l$ by $\Delta^\theta_l$, so that
\begin{align}
  \label{eq:Etheta}
(E^\theta_l)^2 &= \alpha_l\eta_l^2 + \beta_l\widetilde{E}_{l-1}^2,\word{and}\\
  \label{eq:Deltatheta}
(\Delta^\theta_l)^2 &= \sigma_l^2 - 2\theta_l\rho_l\sigma_{l-1}\sigma_l + \theta_l^2\sigma_{l-1}^2.
\end{align}
We will seek to determine \emph{optimal} weights $\widetilde{\theta}_l$ that minimise the cost of computing ${\cal P}^\theta_l$, for a given variance. We denote the resulting  estimators by $\widetilde{\cal P}_l$ and $\widetilde{Y}_l$, and the corresponding values of $E^\theta_l$ and $\Delta^\theta_l$ by $\widetilde{E}_l$ and $\widetilde{\Delta}_l$, respectively.
\begin{proposition}
\label{prop:WMLMC}
The  values for the weights $\widetilde{\theta}_l$ and the effort parameters $\alpha_l$ and $\beta_l$ that generate estimators $\widetilde{\cal P}_l$ with variance $v^2$ for the minimal cost satisfy the following equations, which may be solved recursively:
\[ \widetilde{\Delta}_0 = \sigma_0,\; \theta_0 = 0,\; \widetilde{E}_0 = \frac{\sigma_0\eta_0}{v},\; \alpha_0 = \frac{\sigma_0^2}{v^2},\; \beta_0 = 0,\]
and, for $l>0$, if\; $\abs{\rho_l}>\frac{v\widetilde{E}_{l-1}}{\sigma_{l-1}\eta_l}$, 
\begin{align}
\label{eq:WMLMCthm1}\widetilde{\Delta}_l &= \frac{\sigma_l\sqrt{1-\rho^2_l}}{\sqrt{1-\frac{v^2\widetilde{E}^2_{l-1}}{\sigma^2_{l-1}\eta^2_l}}} \\
\label{eq:WMLMCthm2}   \widetilde{\theta}_l & = 
\frac{\rho_l\sigma_l}{\sigma_{l-1}}-\sgn\rho_l\; \frac{\widetilde{\Delta}_lv\widetilde{E}_{l-1}}{\sigma^2_{l-1}\eta_l}\\
\label{eq:WMLMCthm3}\widetilde{E}_l &= \frac{1}{v}\left(\widetilde{\Delta}_l\eta_l + \abs{\widetilde{\theta}_l}\widetilde{E}_{l-1}v\right)  \\
\label{eq:WMLMCthm4}\alpha_l &= \frac{\widetilde{E}_l\widetilde{\Delta}_l}{\eta_lv}, \qquad
\beta_l = \frac{\widetilde{E}_l\abs{\widetilde{\theta}_l}}{\widetilde{E}_{l-1}},
\end{align}
and otherwise 
\begin{equation} 
\widetilde{\Delta}_l = \sigma_l,\; \widetilde{\theta}_l = 0,\; \widetilde{E}_l = \frac{\sigma_l\eta_l}{v},\; \alpha_l = \frac{\sigma_l^2}{v^2},\; \beta_l = 0.
\end{equation}
\end{proposition}
\begin{proof}
We start by supposing that we have an optimised estimator $\widetilde{\cal P}_{l-1}$ at level $l-1$ with variance $v^2$ and computational cost $\widetilde{E}_{l-1}^2$. From \eqref{eq:MLMCW}, for any given value of $\theta_l$ we have
\[ {\cal P}^\theta_l = \ps{\alpha_l}{Y^\theta_l}+\theta_l\,\ps{\beta_l\,}{\widetilde{\cal P}_{l-1}}.\]
The cost of computing ${\cal P}^\theta_l$ is given by \eqref{eq:Etheta},
and the variance of ${\cal P}^\theta_l$ is given by
\[ \V[{\cal P}^\theta_l] = \frac{(\Delta_l^\theta)^2}{\alpha_l} + \frac{v^2\theta_l^2}{\beta_l}.\]
Minimising $(E^\theta_l)^2$ subject to $\V[{\cal P}^\theta_l]=v^2$ (by means of a Lagrange multiplier, for example) gives us provisional versions of \eqref{eq:WMLMCthm3} and \eqref{eq:WMLMCthm4}:
\begin{equation}
\label{eq:EthetaOpt}
E^\theta_l = \frac{1}{v}\left(\Delta^\theta_l\eta_l + \abs{\theta}_l\widetilde{E}_{l-1}v\right),
\word{with}
\alpha_l = \frac{E^\theta_l\Delta^\theta_l}{\eta_lv}, \word{and}
\beta_l = \frac{E^\theta_l\abs{\theta_l}}{\widetilde{E}_{l-1}}
.
\end{equation}
It remains to determine the value $\widetilde\theta_l$ that minimises $E^\theta_l$ in \eqref{eq:EthetaOpt}. If $\theta_l\ne 0$, we can differentiate the expression for $E^\theta_l$ in \eqref{eq:EthetaOpt}, using \eqref{eq:Deltatheta}, to obtain
\[ \frac{d E^\theta_l}{d \theta_l} = \frac{\eta_l\sigma_{l-1}(\theta_l\sigma_{l-1}-\rho_l\sigma_l)}{v\Delta^\theta_l} + \sgn\theta_l \widetilde{E}_{l-1}. \]
We note that $\displaystyle{\lim_{\theta_l\to 0+}\frac{d E^\theta_l}{d \theta_l} = \widetilde{E}_{l-1}-\frac{1}{v}\rho_l\eta_l\sigma_{l-1}}$. If $\rho_l>\frac{v\widetilde{E}_{l-1}}{\sigma_{l-1}\eta_l}$ this is negative, and we will find a minimum of $E^\theta_l$ for some $\theta_l>0$. Similarly, $\displaystyle{\lim_{\theta_l\to 0-}\frac{d E^\theta_l}{d \theta_l} = -\widetilde{E}_{l-1}-\frac{1}{v}\rho_l\eta_l\sigma_{l-1}}$. If $\rho_l<-\frac{v\widetilde{E}_{l-1}}{\sigma_{l-1}\eta_l}$ this is positive, and we will find a minimum of $E^\theta_l$ for some $\theta_l<0$. Otherwise, the minimum of $E^\theta_l$ will occur when $\theta_l=0$. In this case, we effectively make level $l$ our new coarsest level.

When $\abs{\rho_l}>\frac{v\widetilde{E}_{l-1}}{\sigma_{l-1}\eta_l}$, equating $\frac{d E^\theta_l}{d \theta_l}$ to zero, moving the second term to the right hand side, multiplying through by $v\Delta^\theta_l$ and squaring gives, at the optimal value,
\[ \eta_l^2\sigma_{l-1}^2\bigl(\widetilde{\Delta}_l^2 - \sigma_l^2(1-\rho_l^2)\bigr) = v^2\widetilde{E}^2_{l-1}\widetilde{\Delta}_l^2.\]
Since $\frac{v\widetilde{E}_{l-1}}{\sigma_{l-1}\eta_l}<\abs{\rho_l}\leq 1$ we can solve this equation to obtain \eqref{eq:WMLMCthm1}. We can also rearrange the equation to obtain \eqref{eq:WMLMCthm2}, noting that the sign of $\widetilde{\theta}_l$ agrees with the sign of $\rho_l$.
\end{proof}
We summarize the resulting WMLMC estimator in the following corollary to Proposition~\ref{prop:WMLMC}, which is proven straightforwardly by induction.
\begin{corollary}
\label{cor:WMLMC}
The weighted multilevel Monte Carlo (WMLMC) estimator at level $L$ is given by
\begin{equation}
  \label{eq:WMLMC}
  \widetilde{\cal P}_L = \sum_{l=0}^L\widetilde\Theta^L_{l}\; \ps{N^L_{l}\,}{\widetilde{Y}_{l}},\quad l=0,\dots,L
\end{equation}
with 
\begin{equation}
\label{eq:Ytilde}
\widetilde{Y}_0 = P_0\word{and}\widetilde{Y}_l = P_l - \widetilde{\theta}_lP^l_{l-1}, \quad l=1,\dots,L,
\end{equation}
where $\widetilde\theta_l$ is defined in \eqref{eq:WMLMCthm2}, and with
with
\begin{equation}
\label{eq:MLMCWThetaTilde}
\widetilde\Theta^l_l = 1, 0\leq l\leq L, \word{and} 
\widetilde\Theta^L_{l} = \prod_{k=l+1}^L\widetilde\theta_k,\quad 0\leq l< \leq L.
\end{equation}
The number of samples to be used at each level is given by
\begin{equation}
  \label{eq:WMLMCN}
  N^L_{l} = \frac{\widetilde{E}_L\widetilde{\Delta}_l\abs{\widetilde\Theta^L_l}}{v\eta_l},\quad 0\leq l\leq L,
\end{equation}
with $\widetilde{E}_l$ and $\widetilde{\Delta}_l$ as given in Proposition~\ref{prop:WMLMC}. Moreover, 
\begin{equation}
  \label{eq:WMLMCE}
  \widetilde{E}_l = \frac{1}{v}\sum_{k=0}^l\abs{\widetilde\Theta^l_k}\widetilde\Delta_k\eta_k, \quad 0\leq l\leq L.
\end{equation}
\end{corollary}

To help us explore the implications of Proposition~\ref{prop:WMLMC},
we introduce notation for the \emph{normalised cost}
\begin{equation}
\label{eq:WMLMCnc} \widetilde{\delta}_l = \frac{v\widetilde{E_l}}{\sigma_l\eta_l},
\end{equation}
which can be seen to be the ratio between $\widetilde{E}_l$ and the (square root) single-level cost $\sigma_l\eta_l/v$ (again assuming\footnote{This assumption can be relaxed, adding a constant $\lambda$ (say) that is the ratio of the cost of generating a sample of $Y^\theta_l$ and the cost of generating a single-level sample $P_l$. This complicates things a little but does not fundamentally change any of the results presented here.} that the cost of generating a sample of $Y^\theta_l$ is the same as the cost of generating a single-level sample $P_l$, even when $\theta_l\ne 0$). 
\begin{corollary}
Suppose that the values of $\widetilde{\theta}_l$ are defined as in Proposition~\ref{prop:WMLMC}, and that $\widetilde{\delta}_l$ is as defined in \eqref{eq:WMLMCnc}. Then $\widetilde{\delta}_0 =1$, and, for $l=1,2,\dots$,
\begin{equation}
\label{eq:WMLMCdelta}
\widetilde{\delta}_l = \begin{cases}
\mu_l\abs{\rho_l}\widetilde{\delta}_{l-1} + \sqrt{1-\rho_l^2}\sqrt{1-\mu^2_l\widetilde{\delta}_{l-1}^2} & \word{if} \abs{\rho_l}> \mu_l\widetilde{\delta}_{l-1}\\
1 & \word{otherwise,}
\end{cases}
\end{equation}
where  $\mu_l = \eta_{l-1}/\eta_l$.
\end{corollary}
\begin{proof}
  The result follows from substituting the optimal value of $\theta_l$ from \eqref{eq:WMLMCthm2} into \eqref{eq:WMLMCthm3}.
\end{proof}
We see immediately from \eqref{eq:WMLMCdelta} that the cost of the WMLMC estimator is bounded from above by the single level cost, so that $\widetilde{\delta}_l\leq 1$, and that it is strictly lower than that cost whenever $\abs{\rho_l}>\mu_l$. Moreover, since each MLMC estimator ${\cal P}_L$ is a particular case of the weighted estimator ${\cal P}^\theta_L$ with each $\theta_l=1$, the cost of the optimally weighted estimator $\widetilde{\cal P}_L$ will be no greater than the cost of the MLMC estimator. This implies that we can inherit the results of the MLMC Complexity Theorem (Theorem~\ref{thm:MLMCComplexity}). Nevertheless, we now give a direct statement of the corresponding result for the optimally weighted MLMC method.

\subsection{Complexity theorem}
\begin{lemma}
\label{lem:WMLMCComplexity}
Let $L>0$. Suppose that the values of $\widetilde{\theta}_l$ are defined as in Proposition~\ref{prop:WMLMC},  that $\widetilde{\delta}_l$ is as defined in \eqref{eq:WMLMCnc}, and further that $\abs{\rho_{l}}>\mu_{l}$, $l=1,2,\dots,L$. Then the computational effort required to compute $\widetilde{\cal P}_L$, with variance $v^2$, is $\widetilde{E}_L^2$, where
\begin{equation}
\widetilde{E}_L \leq \frac{\sigma_L}{v}\sum_{l=1}^L\eta_{l}\sqrt{1-\rho^2_{l}}.
\end{equation}
\end{lemma}
\begin{proof}
From \eqref{eq:WMLMCdelta}, we see immediately that
\[ \widetilde{\delta}_l\leq \mu_l\widetilde{\delta}_{l-1}+\sqrt{1-\rho_l^2}. \]
Then, since $\widetilde{\delta}_0=1$, and given that $\mu_l=\eta_{l-1}/\eta_l$, we can deduce by induction that
\[ \widetilde{\delta}_L\leq \frac{1}{\eta_L}\sum_{l=1}^L\eta_{l}\sqrt{1-\rho^2_{l}}.\]
The result follows.
\end{proof}
Lemma~\ref{lem:WMLMCComplexity} can be used to prove the following result. The proof is similar to the proof of Theorem~\ref{thm:MLMCComplexity} and will not be given here.
\begin{theorem}
\label{thm:WMLMCComplexity}
Suppose that there are positive constants $c_1, c_2, c_3, \alpha,\beta,\gamma$ with $2\alpha\geq\min(\beta,\gamma)$ such that, for $l=0,1,2,\dots$,
\begin{enumerate}
\item $\abs{\overline{P}_l-\overline{P}}\leq c_1 2^{-\alpha l} $,
\item $\sqrt{1-\rho^2_l} \leq c_2 2^{-\beta l}$, and
\item $\eta_l \leq c_3 2^{\gamma l}$.
\end{enumerate}
Then there exists a positive constant $c_4$ such that, for any $\epsilon<1/e$ there exists an integer $L$ such that $\widetilde{\cal P}_L$ (as defined by \eqref{eq:MLMCP}--\eqref{eq:MLMCN}) has mean squared error with bound
\[ \E[(\widetilde{\cal P}_L-\overline{P})^2] < \epsilon^2\]
with a computational complexity $C$ with bound
\[ C \leq \begin{cases}
c_4\epsilon^{-2} & \beta>\gamma,\\
c_4\epsilon^{-2}(\log\epsilon)^2 & \beta=\gamma,\\
c_4\epsilon^{-2-\frac{\gamma-\beta}{\alpha}} & \beta<\gamma.
\end{cases}
\]
\end{theorem}
\subsection{Comparison with MLMC}
If we also define $\delta_l = \frac{vE_l}{\eta_l\sigma_l}$, then the MLMC analogue to \eqref{eq:WMLMCdelta} is
\begin{equation}
\label{eq:MLMCdelta}
\delta_l = \begin{cases}
\frac{\sigma_{l-1}}{\sigma_l}\mu_l\delta_{l-1}+\frac{\Delta_l}{\sigma_l} &\word{if}
\rho_l > \mu_l\delta_{l-1}+\frac{\sigma_{l-1}}{2\sigma_l}\left( 1-\mu_l^2\delta_{l-1}^2\right)
\\
1 & \word{otherwise.}
\end{cases}
\end{equation}
Given that the MLMC estimator is a special case of the weighted MLMC estimator (with $\theta_l\equiv 1$), it is necessarily true that $\widetilde\delta_l\leq {\delta}_l$. This can be verified to be the case directly using \eqref{eq:WMLMCdelta} and \eqref{eq:MLMCdelta}. Indeed, it is possible to show from these equations that, if $\widetilde{\delta}_{l-1}=\delta_{l-1}$, $\widetilde{\delta}_l \leq \delta_l$, with equality when $\abs{\rho_l}>\delta_{l-1}\mu_l$ only if
\[ \frac{\sigma_{l-1}}{\sigma_l} = \rho_l - \mu_l\delta_{l-1}\sqrt{\frac{1-\rho_l^2}{1-\mu_l\delta_{l-1}^2}}. \]
Note that, as $l\to\infty$, we expect that $\sigma_{l-1}/\sigma_{l}\to 1$, and $\rho_l\to 1$, so that we do not anticipate that adding (the optimal) weights will improve the \emph{asymptotic} complexity of the method. 

However, signficant gains are still possible in some settings, mainly because of the potential to make more efficient use of the low cost--high bias--low correlation samples at the coarser levels. To gain some insight into this, we consider the simplified situation where the single-level samples all have the same standard deviation $\sigma$, and we are working on a geometric grid with $M=2$, i.e.~$\mu_l = 1/\sqrt{2}$. We suppose too that we are at the coarsest level, so that $\delta_0=\widetilde{\delta}_0 = 1$.  Then, if $\rho_1>\frac{1}{\sqrt{2}}+\frac14$,
\begin{equation}
\label{eq:twolevelcomparison} \delta_1 = \frac{1}{\sqrt{2}} + \sqrt{2(1-\rho_1)},\quad \rho_1>\frac{1}{\sqrt{2}}+\frac14, \word{and} \widetilde{\delta}_1 = \frac{1}{\sqrt{2}}\left(\rho_1 + \sqrt{1-\rho_1^2}\right),\quad \rho_1>\frac{1}{\sqrt{2}}.
\end{equation}
This highlights the comparison with the weighted version derived from \eqref{eq:WMLMCdelta}. Again there are two terms, but the weighted version has an additional factor of $\abs{\rho_l}$ multiplying the first term, and an additional factor of $\frac12 \sqrt{1+\rho_l}$ multiplying the second term, and can be applied for a wider range of values of $\rho_1$. 

\begin{figure}
\centerline{\includegraphics[width=0.7\textwidth]{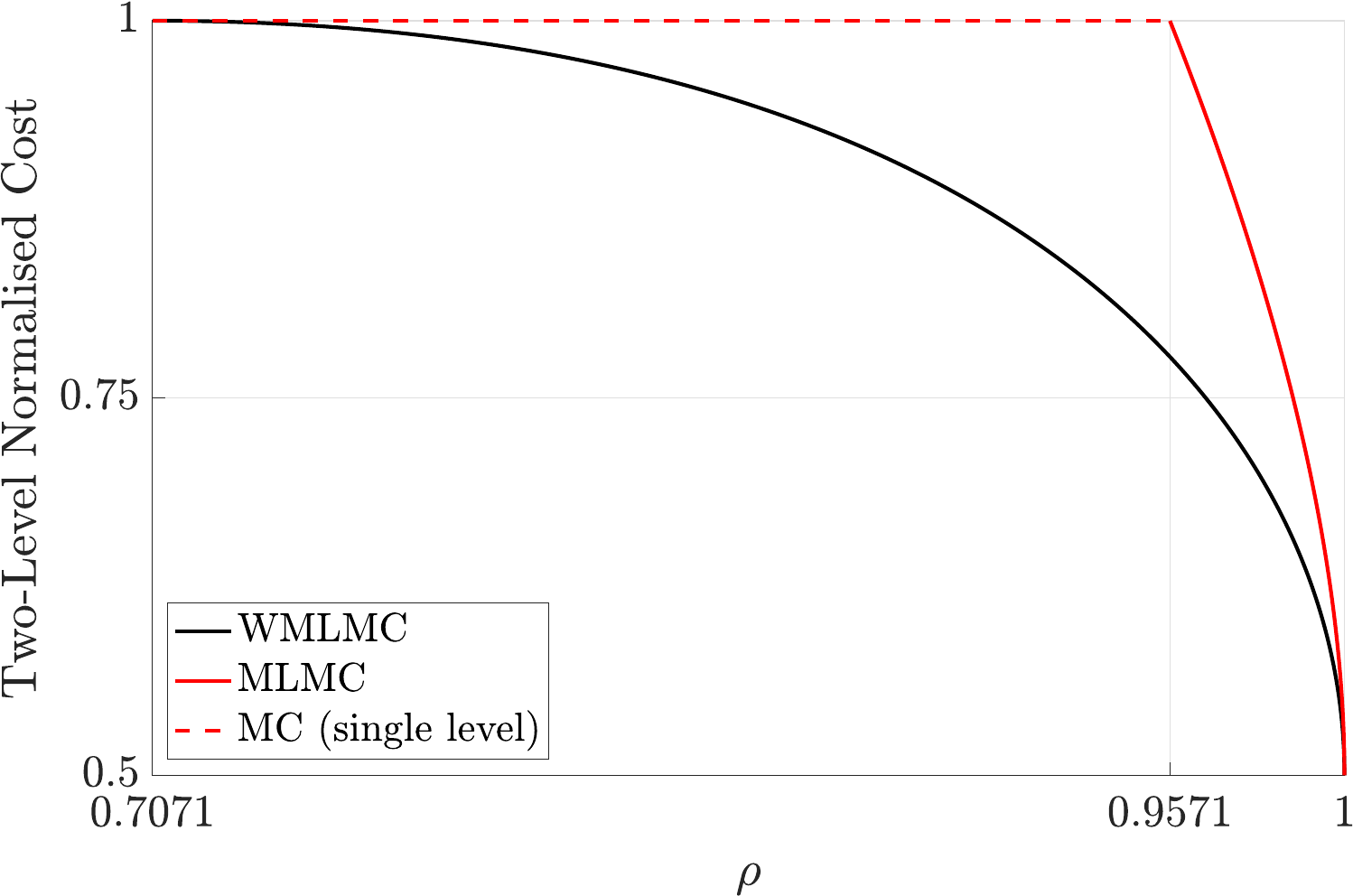}}
\caption{The normalised costs $\widetilde{\delta}^2_1$ (for the WMLMC estimator) and $\delta^2_1$ (for the MLMC estimator) with two levels, as a function of the correlation $\rho$, when $\sigma_0=\sigma_1$, and with $M=2$. Note that, for $\rho<\frac{1}{\sqrt{2}}+\frac14$, the MLMC estimator is more costly than the single-level MC estimator. The maximum ratio between these costs is $1.2865$ and occurs when $\rho=\frac{1}{\sqrt{2}}+\frac14$.}
\label{fig:twolevelcomparison}
\end{figure}

The comparative costs can be seen in Figure~\ref{fig:twolevelcomparison}, which shows $\widetilde{\delta}^2_1$ and $\delta^2_1$ (from \eqref{eq:twolevelcomparison}) as functions of the correlation $\rho$ between $P^1_0$ and $P_0$, for $\rho\geq \frac{1}{\sqrt{2}}$.  The maximum ratio between these costs is $1.2865$, which occurs when $\rho=\frac{1}{\sqrt{2}}+\frac14$. At this point, $\delta_1=1$, and for lower values of the correlation, the single-level MC estimate (with a normalised cost of $1$) is more efficient than the two-level MLMC estimate.

\begin{figure}
\centerline{\includegraphics[width=0.9\textwidth]{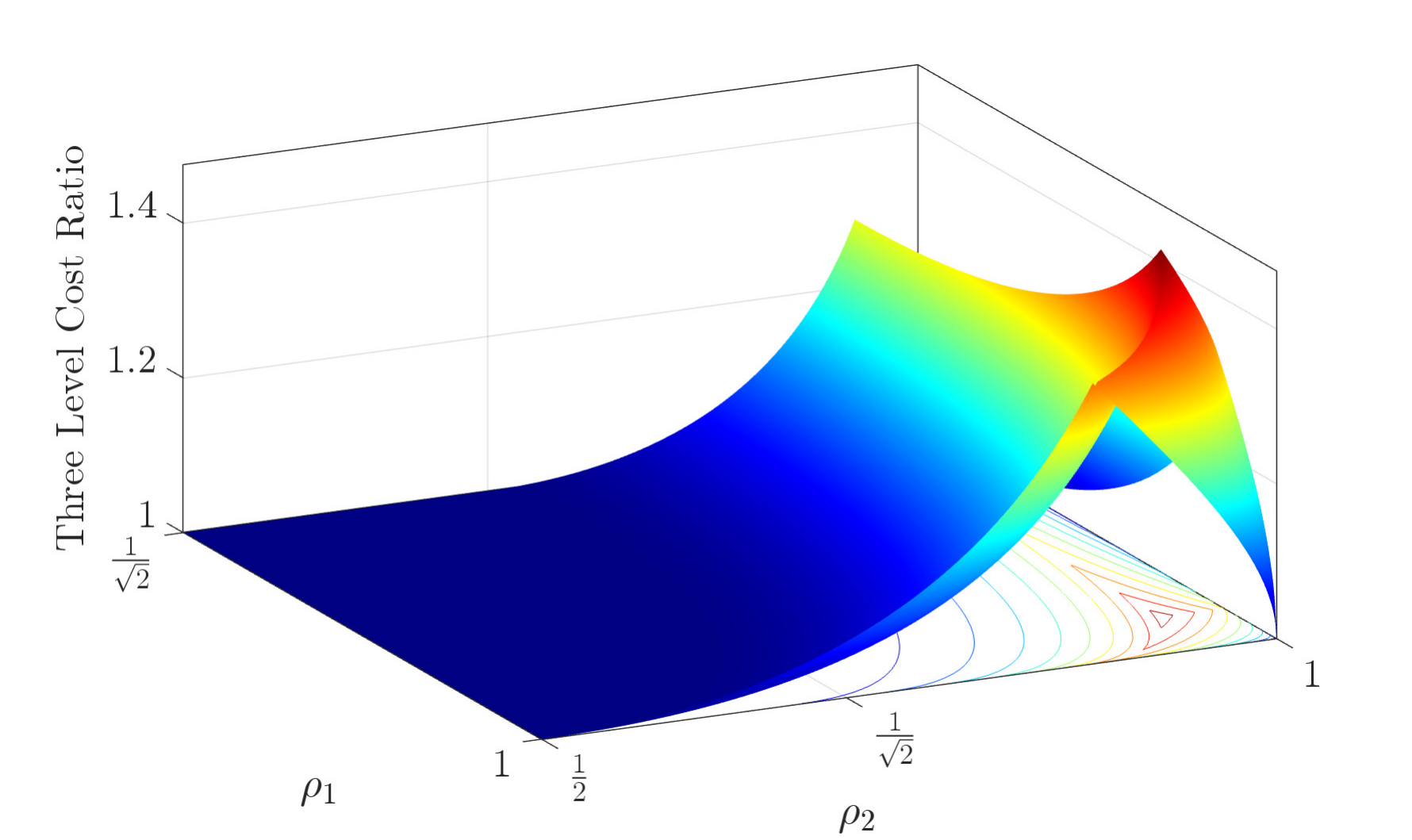}}
\caption{The ratio  $\delta_2^2/\widetilde{\delta}^2_2$ of normalised costs between the MLMC and WMLMC estimators, with three levels, shown as a function of the correlations $\rho_1$ and $\rho_2$, with $\sigma_0=\sigma_1=\sigma_2$, and with $M=2$. The maximum ratio between these costs is $1.4752$ and occurs when $\rho_1=\rho_2=\frac{1}{\sqrt{2}}+\frac14$.}
\label{fig:threelevelcomparison}
\end{figure}

In Figure~\ref{fig:threelevelcomparison} we show the three-level comparison in the form of the ratio between $\delta_2^2$ and $\widetilde{\delta}_2^2$, as a function of the two correlations $\rho_1$ and $\rho_2$, between $P^1_0$ and $P_0$ and between $P^2_1$ and $P_1$, respectively. Note that the potential efficiency savings are compounded, and are maximised when the correlations are close to $0.9571$.

\section{Weighted Multi Index Multilevel Monte Carlo}
In this section we show that a weighted version of the multi-index multilevel Monte Carlo (MIMC) formulation introduced by \cite{Haji-AliNobileTempone2015} can also be formulated. In order to do so, we review the MIMC method and show that it can be written in a recursive form, and weights added.

Our goal is to estimate a quantity $\overline{P}$ via a set of approximate estimators indexed by multi-indices\footnote{For a d-dimensional multi-index $\lambda = (\lambda_1,\dots,\lambda_d)\in \N_0^d$, we define $\abs{\lambda} = \sum_i \lambda_i$. For two multi-indices $\lambda$ and ${\lambda'}$, we say that ${\lambda'}\leq\lambda$ if ${\lambda'}_i\leq\lambda_i$, $i=1,\dots,d$, and ${\lambda'}<\lambda$ if ${\lambda'}\leq\lambda$ and ${\lambda'}\ne\lambda$. We also define $\lambda\wedge{\lambda'}$ to be the multi-index with $i$th entry $\min\{\lambda_i,{\lambda'}_i\}$, and we define $\lambda\vee{\lambda'}$ similarly. We also define subsets $\square^-_\lambda = \{{\lambda'}\geq {\bf 0}: {\bf 0}<\lambda-{\lambda'}\leq {\bf 1}\}$ and $\square^+_\lambda = \{{\lambda'}\geq {\bf 0}: {\bf 0}<{\lambda'}-\lambda\leq {\bf 1}\}$.} $\lambda\in \N_0^d$. These have means $\overline{P}_\lambda$, with $\overline{P}_\lambda\to\overline{P}$ as $\min_i \{\lambda_i\}\to\infty$.

A single-level estimator for $\overline{P}_\lambda$ will be denoted by $P_\lambda$. Whenever a sample of $P_\lambda$ is computed, we also compute $\{P^\lambda_{\lambda'}\}_{{\lambda'}\in\square^-_\lambda}$, a collection of estimators for the lower index quantities $\{\overline{P}_{\lambda'}\}_{{\lambda'}\in\square^-_\lambda}$, but computed from the same underlying random variable(s). These are collected  into basic estimators 
\[ Y_\lambda = P_\lambda - \sum_{{\lambda'}\in\square^-_\lambda}\epsilon^\lambda_{\lambda'} P^\lambda_{\lambda'},\]
where $\epsilon^\lambda_{\lambda'}: = (-1)^{1+\abs{\lambda-{\lambda'}}}$. We denote the computational effort needed to compute one sample of $Y_\lambda$ by $\eta_\lambda^2$, and by $\Delta^2_\lambda$ the variance of $Y_\lambda$. 

A multilevel estimator ${\cal P}_\Lambda$ at level $\Lambda$ is then constructed from the basic estimators $Y_\lambda$ in the form
\begin{equation}
    \label{eq:MIMLMCP}
    {\cal P}_\Lambda = \sum_{{\bf 0}\leq{\lambda}\leq \Lambda}{_{N^\Lambda_{\lambda}}}Y_{\lambda}
\end{equation} 
where ${_{N^\Lambda_{\lambda}}}Y_{\lambda}$ denotes the average of ${_{N^\Lambda_{\lambda}}}$ independent samples of $Y_{\lambda}$. It is straightforward\footnote{We note that 
\[\E[{\cal P}_\Lambda] =\sum_{{\bf 0}\leq{\lambda}\leq \Lambda}\left(\overline{P}_{\lambda} - \sum_{{\lambda'}\in\square^-_{\lambda}}\epsilon^{\lambda}_{\lambda'} \overline{P}_{\lambda'}\right) = \sum_{{\bf 0}\leq{\lambda}\leq \Lambda}\overline{P}_{\lambda}\left(1 - \sum_{{\lambda'}\in\square^+_{\lambda},\;\lambda'\leq\Lambda}\epsilon^{\lambda'}_{\lambda} \right), \] and that the term in brackets is zero unless $\lambda=\Lambda$, in which case the sum inside the bracket is empty and we are left with $\overline{P}_{\Lambda}$.} 
to verify that $\E[{\cal P}_\Lambda] = \overline{P}_\Lambda$, and that, since the $Y_\lambda$ are computed independently of one another, $\V[P_\Lambda] = \sum_{{\bf 0}\leq{\lambda}\leq\Lambda}\frac{\Delta_\lambda^2}{N^\Lambda_{\lambda}}$. The total computational effort required to compute ${\cal P}_\Lambda$ is $W^2_\Lambda=\sum_{{\bf 0}\leq{\lambda}\leq \Lambda}N^\Lambda_{\lambda}\eta^2_{\lambda}$.

We set a target variance for ${\cal P}_\Lambda$ of $v^2$, and seek to find the least computationally expensive way of achieving that. Minimising $W^2_\Lambda$, subject to $\V[{\cal P}_\Lambda]=v^2$, results in 
\begin{equation}
\label{eq:MIMLMCdeltaW}
    W_\Lambda = \frac{1}{v}\sum_{{\bf 0}\leq{\lambda}\leq\Lambda}\Delta_{\lambda}\eta_{\lambda},
\word{and} 
N^\Lambda_{\lambda} = \frac{W_{\Lambda}\Delta_{\lambda}}{v\eta_{\lambda}},\quad 0\leq\lambda\leq\Lambda.
\end{equation}
\subsection{A recursive formulation}
As we saw in the single-index case, the first step to formulating a weighted version of the method is to write it in recursive form. To wit, we have the following proposition.
\begin{proposition}
The multilevel estimators ${\cal P}_\lambda$, $\lambda\geq {\bf 0}$ defined in \eqref{eq:MIMLMCP} and \eqref{eq:MIMLMCdeltaW} satisfy the recursive relation
\begin{equation}
\label{eq:recursiveP}
    {\cal P}_\lambda = {_{\alpha_\lambda}}Y_\lambda + \sum_{{\lambda'}\in\square^-_\lambda}\epsilon^\lambda_{\lambda'}\;{_{\frac{\beta_\lambda}{\beta_{\lambda'}}}}{\cal P}_{\lambda'},
\end{equation}
where $\alpha_\lambda = N^\lambda_\lambda$ and $\beta_{\lambda} = W_\lambda$, $\lambda\geq{\bf 0}$.
\end{proposition}
\begin{proof}
The proof is by induction on $\lambda$. Let $\kappa>{\bf 0}$, and suppose that \eqref{eq:recursiveP} holds for all $\lambda\in\square^-_\kappa$. Then, using \eqref{eq:MIMLMCP},
\begin{align*}
{_{\alpha_\kappa}}Y_\kappa + 
\sum_{{\lambda}\in\square^-_\kappa}\epsilon^\kappa_{\lambda}\;{_{\beta^\kappa_{\lambda}}}{\cal P}_{\lambda}
&= {_{N^\kappa_\kappa}}Y_\kappa + \sum_{{\lambda}\in\square^-_\kappa}\epsilon^\kappa_{\lambda}\;\sum_{0\leq\lambda'\leq\lambda}{_{\beta^\kappa_{\lambda}N^\lambda_{\lambda'}}}Y_{\lambda'}\\
&= {_{N^\kappa_\kappa}}Y_\kappa + \sum_{{\lambda}\in\square^-_\kappa}\epsilon^\kappa_{\lambda}\;\sum_{0\leq\lambda'\leq\lambda}{_{N^\kappa_{\lambda'}}}Y_{\lambda'}\\
&= {_{N^\kappa_\kappa}}Y_\kappa +\sum_{0\leq\lambda'<\kappa}\left( \sum_{{\lambda}\in\square^-_{\kappa},\lambda\geq\lambda'}\epsilon^\kappa_{\lambda}\right)\;{_{N^\kappa_{\lambda'}}}Y_{\lambda'}\\
&= {_{N^\kappa_\kappa}}Y_\kappa +\sum_{0\leq\lambda'<\kappa}\;{_{N^\kappa_{\lambda'}}}Y_{\lambda'} = {\cal P}_\kappa,
\end{align*}
so that \eqref{eq:recursiveP} holds for $\lambda = \kappa$, as required. The inductive proof is completed by noting that the $\lambda={\bf 0}$ case holds because in that case the sum on the right hand side of \eqref{eq:recursiveP} is empty.
\end{proof}

\subsubsection*{Adding weights}
As in the single-index case, we define weighted versions of the multilevel estimators. We start with local weights $\theta^\lambda_\nu$ for $\nu\in\square^-_\lambda$, which we take to be given (for the moment), and we define
\begin{equation}
    \label{eq:Y}
    Y^\theta_\lambda = P_\lambda - \sum_{\nu\in\square^-_\lambda}\theta^\lambda_\nu P^\lambda_\nu
\end{equation}
and 
\begin{equation}
    \label{eq:P}
    {\cal P}^\theta_\lambda = \ps{\alpha_\lambda}{Y^\theta_\lambda} + \sum_{\nu\in\square^-_\lambda}\theta^\lambda_\nu\ps{\frac{\beta_\lambda}{\beta_\nu}}{{\cal P}^\theta_\nu},
\end{equation}
where the effort parameters $\alpha_\lambda$ and $\beta_\lambda$ are to be chosen to minimise the effort required to compute ${\cal P}^\theta_\lambda$ subject to $\V[{\cal P}^\theta_\lambda] = v^2$.

We are going to look to expand \eqref{eq:P} as a collection of basic estimators $Y_\nu$. Doing this will generate sums of products of the weights $\theta^\lambda_\nu$. We denote these by $\Theta^\lambda_\nu$ for ${\bf 0}\leq\nu,\lambda$, which we define recursively:
\begin{equation}
    \label{eq:Theta}
    \Theta^\lambda_\nu = \begin{cases}
        1 & \word{if} \lambda=\nu\\
        \sum_{\kappa\in\square^-_\lambda}\theta^\lambda_\kappa\Theta^\kappa_\nu & \word{if} \nu <\lambda\\
        0 & \word{otherwise.}
    \end{cases}
\end{equation}
Intuitively, $\Theta^\lambda_\nu$ is the sum of products of the weights $\theta^\kappa_{\kappa'}$ along all possible paths from $\nu$ to $\lambda$. The definition above can be used to recursively define these sums of products, starting from ${\bf 0}$. 

We can now expand \eqref{eq:P}, as follows.
\begin{proposition}
    Suppose that the estimates ${\cal P}^\theta_\lambda$ satisfy the recurrence relation \eqref{eq:P} for ${\bf 0}\leq\lambda$. Then we may write
    \begin{equation}
\label{eq:Pexpanded}
    {\cal P}^\theta_\lambda = \ps{\alpha_\lambda}{Y^\theta_\lambda} +
\sum_{{\bf 0}\leq\nu<\lambda}\Theta^\lambda_\nu \; 
\ps{\frac{\alpha_\nu\beta_\lambda}{\beta_\nu}}{Y^\theta_\nu}.
\end{equation}
\end{proposition}
\begin{proof}
    We start by supposing that \eqref{eq:Pexpanded} holds for all $\nu<\lambda$. We can then expand the right hand side of \eqref{eq:P}, using \eqref{eq:Pexpanded}, to obtain
    \begin{align*}
        \ps{\alpha_\lambda}{Y^\theta_\lambda} + 
\sum_{\nu\in\square^-_\lambda}\theta^\lambda_\nu\;\ps{\frac{\beta_\lambda}{\beta_\nu}}{{\cal P}^\theta_\nu} 
&= \ps{\alpha_\lambda}{Y^\theta_\lambda} + 
\sum_{\nu\in\square^-_\lambda}\theta^\lambda_\nu\sum_{{\bf 0}\leq\kappa<\nu}\Theta^\nu_\kappa
\;\ps{\frac{\alpha_\kappa\beta_\lambda}{\beta_\kappa}}{Y^\nu_\kappa}
\\ 
&= \ps{\alpha_\lambda}{Y^\theta_\lambda} + 
\sum_{{\bf 0}\leq\kappa<\nu}\left(\sum_{\nu\in\square^-_\lambda}\theta^\lambda_\nu\Theta^\nu_\kappa\right)
\;\ps{\frac{\alpha_\kappa\beta_\lambda}{\beta_\kappa}}{Y^\nu_\kappa}\\
&= \ps{\alpha_\lambda}{Y^\theta_\lambda} +
\sum_{{\bf 0}\leq\nu<\lambda}\Theta^\lambda_\nu \; 
\ps{\frac{\alpha_\nu\beta_\lambda}{\beta_\nu}}{Y^\theta_\nu}.
        \end{align*}
\end{proof}
Note that the estimators ${\cal P}^\theta_\nu$ on the right hand side of \eqref{eq:P} are combined in such a way as to ensure that the number of samples of $Y_\kappa$ to be computed is consistently defined. Thus at each node $\nu$, $\ps{\frac{\alpha_\nu\beta_\lambda}{\beta_\nu}}{Y^\theta_\nu}$ can be computed, independently of the values at other nodes, using $\frac{\alpha_\nu\beta_\lambda}{\beta_\nu}$ samples\footnote{In practice, of course, the number of samples is rounded to the nearest positive integer.} of $Y^\theta_\nu$.

The work required to compute ${\cal P}^\theta_\lambda$ is
\begin{equation}
\label{eq:WMIMLMCW}
(W^\theta_\lambda)^2 = \alpha_\lambda\eta^2_\lambda + \beta_\lambda\widehat{\eta}^2_\lambda,\word{with}
\widehat{\eta}^2_\lambda = \sum_{{\bf 0}\leq\nu<\lambda}\frac{\alpha_\nu}{\beta_\nu}\eta_\nu^2,
\end{equation}
and, from \eqref{eq:Pexpanded}, 
\begin{equation}
\label{eq:WMIMLMCV}
\V[{\cal P}^\theta_\lambda] = \frac{(\Delta^\theta_\lambda)^2}{\alpha_\lambda}+\frac{(\widehat\Delta^\theta_\lambda)^2}{\beta_\lambda},
\word{with}(\Delta^\theta_\lambda)^2 = \V[Y^\theta_\lambda] \word{and}
(\widehat\Delta^\theta_\lambda)^2 = \V\left[\sum_{{\bf 0}\leq\nu<\lambda}\Theta^\lambda_\nu \; 
\ps{\frac{\alpha_\nu}{\beta_\nu}}{Y^\theta_\nu}\right] = \sum_{{\bf 0}\leq\nu<\lambda}\big(\Theta^\lambda_\nu\Delta^\theta_\nu\big)^2\frac{\beta_\nu}{\alpha_\nu}.
\end{equation}
When $\lambda={\bf 0}$, we can minimise $W^\theta_{\bf 0}$ subject to $\V[P^\theta_{\bf 0}]=v^2$ by setting $\beta_{\bf 0} =\alpha_{\bf 0} = \frac{\sigma_{\bf 0}}{v\eta_{\bf 0}}$. For $\lambda>{\bf 0}$, we assume that $\alpha_\nu$ and $\beta_\nu$ have been determined for all $\nu<\lambda$, and we choose $\alpha_\lambda$ and $\beta_\lambda$ so as to minimise $W^\theta_\lambda$ subject to $\V[{\cal P}^\theta_\lambda] =v^2$. This results in
\begin{equation}
\label{eq:WMIMLMCWab}
W^\theta_\lambda = \frac{1}{v}\big(\eta_\lambda\Delta^\theta_\lambda + \widehat{\eta}_\lambda\widehat{\Delta}^\theta_\lambda\big), \word{with}
\alpha_\lambda = \frac{\Delta^\theta_\lambda W^\theta_\lambda}{v\eta_\lambda}\word{and}
\beta_\lambda = \frac{\widehat\Delta^\theta_\lambda W^\theta_\lambda}{v\widehat{\eta}_\lambda}.
\end{equation}
Note here that the original MIMC method of \cite{Haji-AliNobileTempone2015} corresponds in this context to taking $\theta^\lambda_\nu = \epsilon^\lambda_\nu$. 

It remains to determine the optimal weights $\theta = \{\Theta^\lambda_\nu\}$. We have, by definition, $\Theta^{\bf 0}_{\bf 0} = 1$. We suppose, inductively, that the values of $\Theta^\kappa_\nu$ have been determined for all $\kappa,\nu<\lambda$. We may then expand, from \eqref{eq:WMIMLMCV} and using \eqref{eq:Theta},
\begin{equation}
\label{eq:WMIMLMCR}
(\widehat\Delta^\theta_\lambda)^2 = \sum_{\kappa,\kappa'\in\square^-_\lambda}\theta^\lambda_{\kappa}\theta^\lambda_{\kappa'}R_{\kappa,\kappa'}, \word{where}R_{\kappa,\kappa'}=
\sum_{{\bf 0}\leq\nu\leq\kappa\wedge\kappa'}\Theta^\kappa_\nu\Theta^{\kappa'}_\nu\big(\Delta^\theta_\nu\big)^2\frac{\beta_\nu}{\alpha_\nu}.
\end{equation}
We can write this more succinctly in vector form. We introduce the vectors ${\bf t}_\lambda = [\theta^\lambda_\nu]_{\nu\in\square^-_\lambda}$, ${\bf 0}<\lambda\leq\overline{\lambda}$, allowing us to write
\begin{equation}
\label{eq:DeltaHat}
\widehat\Delta^\theta_\lambda = \sqrt{{\bf t}'_\lambda R_\lambda {\bf t}_\lambda},\word{where}R_\lambda = [R_{\kappa,\kappa'}]_{\kappa,\kappa'\in\square^-_\lambda}.
\end{equation}
We can also express
\begin{equation}
    \label{eq:VarY}
    (\Delta^\theta_\lambda)^2 = \V[Y_\lambda] = 
    \sigma^2_\lambda - 2{\bf c}'_\lambda{\bf t}_\lambda+ {\bf t}'_\lambda C_\lambda {\bf t}_\lambda,
\end{equation} 
where $\sigma_\lambda^2 = \V[P_\lambda]$, ${\bf c}_\lambda=\big[\Cov[P^\lambda_\nu,P^\lambda_\lambda]\big]_{\nu\in\square^-_\lambda}$, and $C_\lambda = \Big[ \Cov[P^\lambda_\nu,P^\lambda_{\nu'}]\Big]_{\nu,\nu'\in\square^-_\lambda}$. 

We can now define the optimal $t_\lambda$ to be the value that minimises
\begin{equation}
\label{eq:thetaopt}
 vW^\theta_\lambda = \eta_\lambda\Delta^\theta_\lambda+\widehat\eta_\lambda\widehat\Delta^\theta_\lambda = \eta_\lambda\sqrt{\sigma^2_\lambda - 2{\bf c}'_\lambda{\bf t}_\lambda+ {\bf t}'_\lambda C_\lambda {\bf t}_\lambda} + \widehat\eta_\lambda\sqrt{{\bf t}'_\lambda R_\lambda {\bf t}_\lambda}.
\end{equation}
It is not, in general, going to be possible to express the optimal value $\widetilde{t}_\lambda$ explicitly, and hence we have no explicit expressions for the corresponding optimal values $\widetilde{W}_\lambda$, $\widetilde{\alpha}_\lambda$, $\widetilde{\beta}_\lambda$, etc. Given estimates for the covariance matrices $\begin{bmatrix}
\sigma_\lambda^2 & {\bf c}'_\lambda\\ {\bf c}_\lambda & C_{\lambda}\end{bmatrix}$, they are, however, readily computable by means of the following sequence of steps, which is continued until $\widetilde{\Theta}^\lambda_\nu$, $\widetilde{\alpha}_\lambda$ and $\widetilde{\beta}_\lambda$ have been found for all ${\bf 0}\leq\nu\leq\lambda\leq\Lambda$, where $\Lambda$ is the highest index at which an estimate is required (c.f.~\eqref{eq:MIMLMCP}). 
\begin{itemize}
\item Initialise the computation by setting $\widetilde{\Delta}_{\bf 0} = \sigma_{\bf 0}$, $\widetilde\beta_{\bf 0} =\widetilde\alpha_{\bf 0} = \frac{\widetilde{\Delta}^2_{\bf 0}}{v^2}$, $\widetilde{W}_{\bf 0} = \frac{\widetilde{\Delta}_{\bf 0}\eta_{\bf 0}}{v}$.
\item Choose ${\bf 0}<\lambda\leq\Lambda$ for which $\widetilde{\Theta}^\kappa_\nu$ have been computed for all ${\bf 0}\leq \nu,\kappa<\lambda$ and $\widetilde{\Delta}_\nu$, $\widetilde{\alpha}_\nu$, $\widetilde{\beta}_\nu$ have been computed for all ${\bf 0}\leq \nu<\lambda$
\begin{itemize}
\item Compute $R_\lambda= [R_{\kappa,\kappa'}]_{\kappa,\kappa'\in\square^-_\lambda}$ using \eqref{eq:WMIMLMCR}.
\item Compute $\widehat\eta_\lambda$ using \eqref{eq:WMIMLMCW}.
\item Find $\widetilde{W}_\lambda$ and the optimising values $\widetilde{\bf t}_\lambda$ by minimising $vW^\theta_\lambda$ in \eqref{eq:thetaopt}.
\item Compute $\widetilde{\Delta}_\lambda$ via \eqref{eq:VarY}, and then $\widetilde{\alpha}_\lambda = \frac{\widetilde{\Delta}_\lambda\widetilde{W}_\lambda}{v\eta_\lambda}$ and $\widetilde{\beta}_\lambda = \frac{\widetilde{W}_\lambda}{v\eta_\lambda}\sqrt{\widetilde{\bf t}'_\lambda R_\lambda \widetilde{\bf t}_\lambda}$.
\item Compute $\widetilde{\Theta}^\lambda_\nu$ for all $\nu<\lambda$ using \eqref{eq:Theta}.
\end{itemize}
\end{itemize}

\section{Numerical experiments}
\label{sect:NumExpts}
In this section we illustrate the performance of the (optimally) weighted (single-index) MLMC method in comparison to MLMC. The implementation of the weighted MIMC method, as well as its performance in a range of settings, will be explored in a subsequent paper (see also \cite{Li2024}).

For the results shown in Figures~\ref{fig:CallEulerGBM2}--\ref{fig:DigitalEulerGBM4}, the covariances, means, optimal weights, etc.~were calculated using $10^6$ samples at each level. When using the methods in practice, the number of samples used would be smaller, and an algorithm along the lines of   Algorithm~1 of \cite{Giles2015} would be used. In Figure~\ref{fig:Histograms} we illustrate numerical results from using such an algorithm, with an initial $20$ samples used at each level to estimate the covariances.

The experiments in this section are based on European option valuation, with three stochastic models (c.f.~\eqref{eq1}) with $a$ and $b$ as shown in Table~\ref{tab:sdes}. We employ two discrete approximations: Euler-Maruyama \eqref{eq:EM} and Milstein \eqref{eq:Mil}:
\begin{equation}
\label{eq:Mil}
    S^l_{j+1}=S^l_j+a(S^l_j)h+b(S^l_j)\Delta W^l_j + \frac12 b(S^l_j)b'(S^l_j)((\Delta W^l_j)^2-h), \hspace{10pt} j=0,1...,J_l-1.
\end{equation}
We also consider three payoff functions, listed in Table~\ref{tab:payoffs}, where we show the discrete forms. In the experiments shown here, we form each of our samples using antithetic variates, so that each SDE path is recomputed using Brownian increments with the opposite sign, the payoffs are computed and the results averaged. We note that, for the Asian payoff, we employ an interpolation-based formula for the intermediate values for $P^l_{l-1}$, as described in Section~5 of \cite{Giles2015}. In  Figure~\ref{fig:DigitalEulerGBM4} we show the results for a digital payoff, and note that here we use the naive implementation (without the payoff smoothing technique described in \cite{Giles2015}). Our purpose here is to illustrate the relative gains possible when adding weights for situations where such smoothing may not be available.

In each of Figures~\ref{fig:CallEulerGBM2}--\ref{fig:DigitalEulerGBM4} we show four graphs. On the left hand side of each figure we show the computational cost (calculated as multiples of the cost of a single sample at level $0$) multiplied by the MSE, as a function of the MSE. The costs for MLMC, WMLMC and single-level MC estimates are shown in the upper figure, while the lower figure shows the ratio between the WMLMC cost and the MLMC cost. On the right hand side of each figure we show the contributions from levels $l=0,\dots,12$ to the total cost of MLMC and WMLMC estimates at level $12$, with the total variance being $0.5\times 10^{-6}$. The cumulative sums of these costs are also shown, marked by triangles, leading to the total costs indicated by solid circles. The bottom right hand graphs show the correlations\footnote{Strictly speaking, the values of $\sqrt{1-\rho_l^2}$ are shown, rather than the correlations themselves.} between $P_l$ and $P^l_{l-1}$, and also the values of $\Theta^L_l$ used for WMLMC.

\begin{table}
\centerline{
\begin{tabular}{r|l}\hline
SDE& Drift and volatility\\
\hline
GBM &$a(S) = 0.05S$,\quad $b(S)=0.2S$\\
IGBM &$a(S) = 2(100-S)$,\quad $b(S)=0.2S$\\
CIR & $a(S) = 2(100-S)$,\quad $b(S)=0.2\sqrt{S}$\\
\hline
\end{tabular}
\vspace*{\baselineskip}
}
\caption{The drift and volatility functions for the SDEs used in the experiments.}
\label{tab:sdes}

\end{table}

\begin{table}
\centerline{
\begin{tabular}{c|c}
\hline
Payoff & Discrete form\\
\hline
Call: $e^{-rT}\max(S_T-100,0)$ & $e^{-rT}\max(S^l_{J_l}-100,0)$\\
Asian: $\displaystyle{\frac{e^{-rT}}{T}\max\left(\int_0^TS_t dt-100,0\right)}$ & 
$\displaystyle{\frac{e^{-rT}}{J_l}\left(\sum_{j=1}^{J_l}S^l_j-100,0\right)}$\\
Digital: $e^{-rT}\max(S_T-100 e^{rT},0)$ &  $e^{-rT}\max(S^l_{J_l}-100 e^{rT},0)$\\
\hline
\end{tabular}
\vspace*{\baselineskip}
}
\caption{The drift and volatility functions for the SDEs used in the experiments.}
\label{tab:payoffs}
\end{table}

\begin{figure}[htb!]
    \centerline{
\begin{minipage}{0.6\linewidth}
    \includegraphics[width=\linewidth]{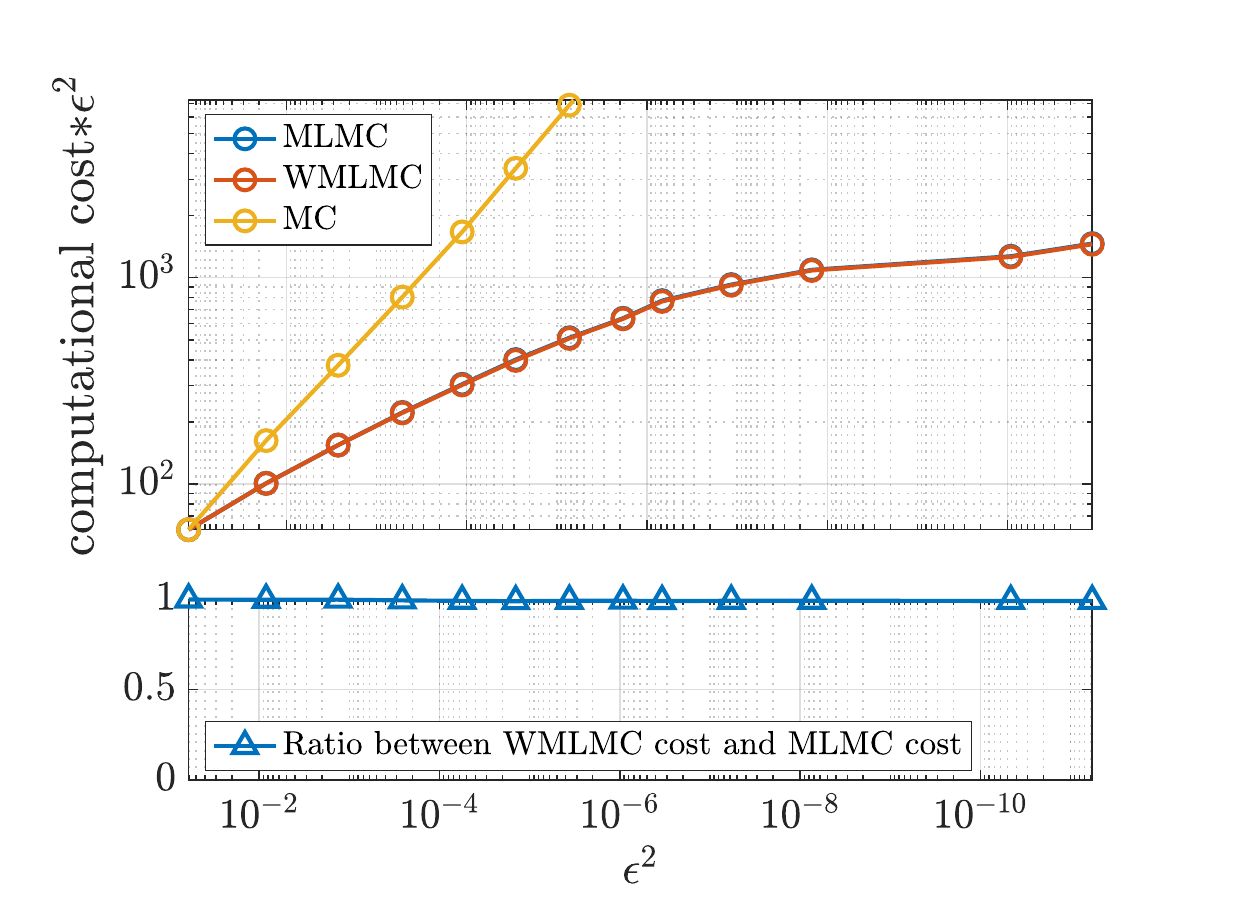}\\
\centerline{(a)}
  \end{minipage}
  \hspace*{-1cm}
\begin{minipage}{0.6\linewidth}
    \includegraphics[width=\linewidth]{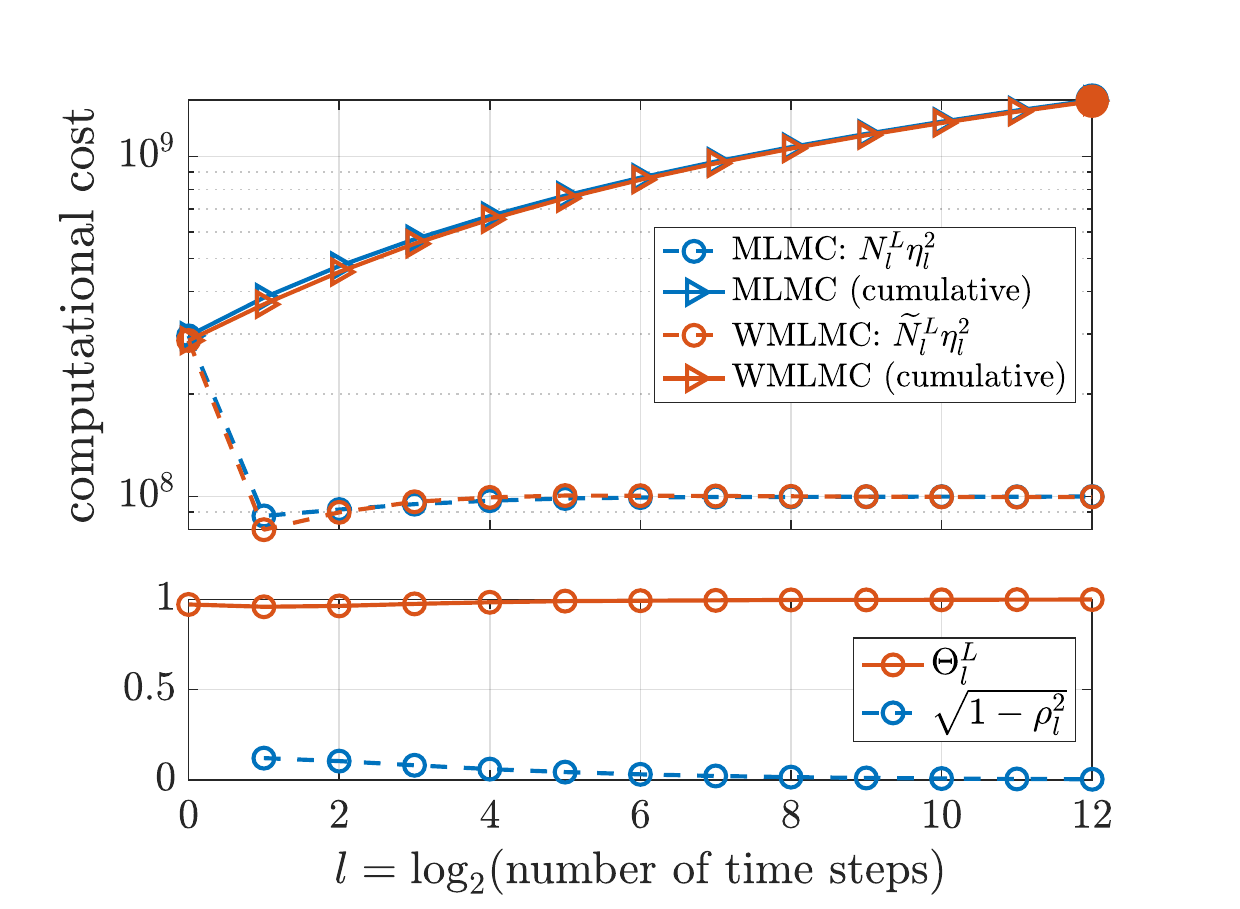}\\
\centerline{(b)}
  \end{minipage}
}
    \caption{European call option under GBM with Euler Maruyama discretisation, $M=2$. (a) Normalised costs for MLMC, WMLMC and single-level Monte Carlo (top) and the ratio between WMLMC and MLMC costs (bottom), shown as functions of the MSE $\epsilon^2$. (b) Contributions at  levels $l=0,\dots,12$ to the cost of WMLMC and MLMC estimates at level $L=12$ (top), with corresponding values of $\Theta^L_l$ and $\sqrt{1-\rho_l^2}$ (bottom). Note that in this case the correlations are near to 1 at all levels, and the costs of the WMLMC and MLMC estimates are virtually identical.} 
\label{fig:CallEulerGBM2}
\end{figure}

The first results we show in Figure~\ref{fig:CallEulerGBM2} serves to show that the WMLMC method does not always offer a measurable improvement over MLMC. In this case, we are evaluating a European call option under GBM with a Euler-Maruyama discretisation, with $M=2$. Note here that the ratio between the MLMC and WMLMC costs shown in the bottom left hand graph is almost exactly 1. 

\begin{figure}[htb!]
    \centerline{
\begin{minipage}{0.6\linewidth}
    \includegraphics[width=\linewidth]{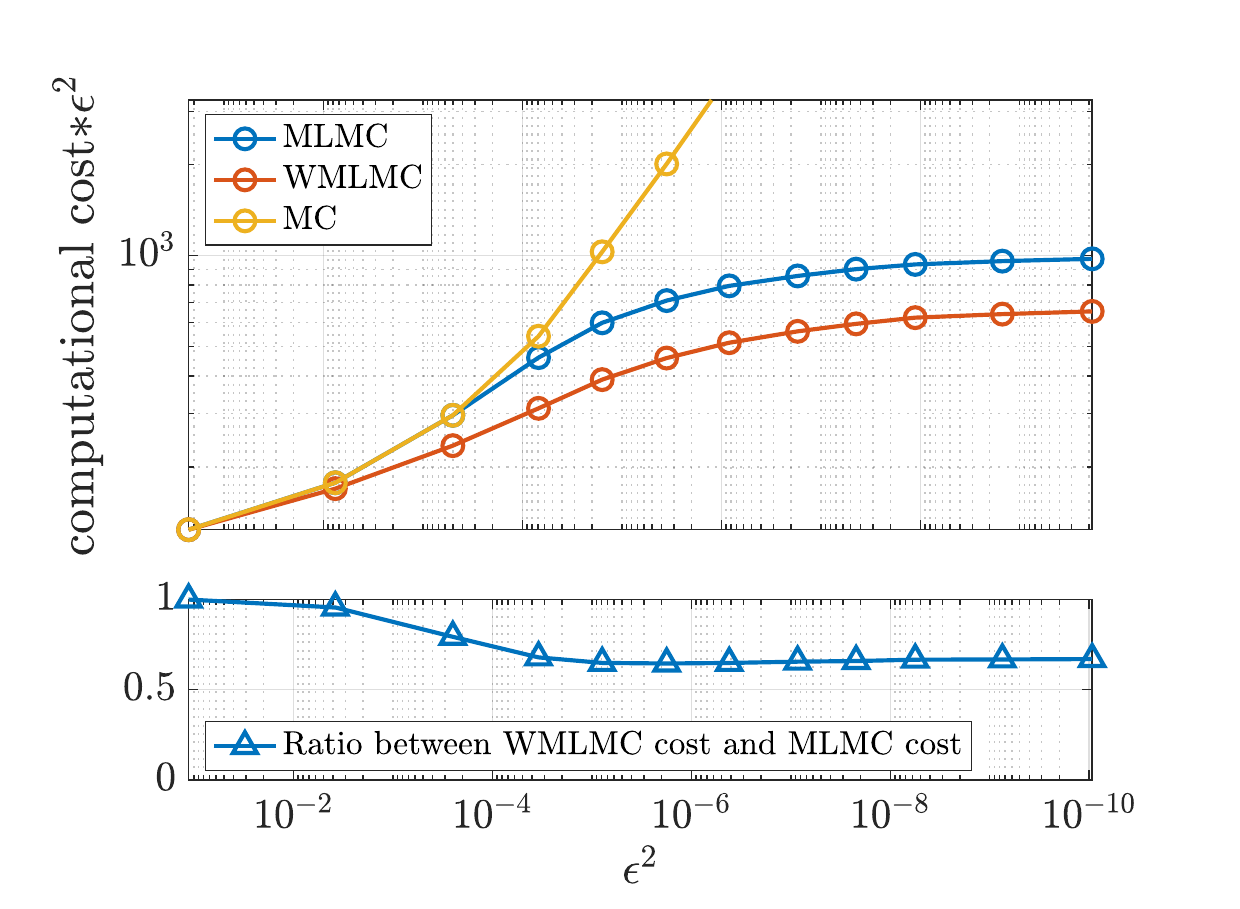}\\
\centerline{(a)}
  \end{minipage}
  \hspace*{-1cm}
\begin{minipage}{0.6\linewidth}
    \includegraphics[width=\linewidth]{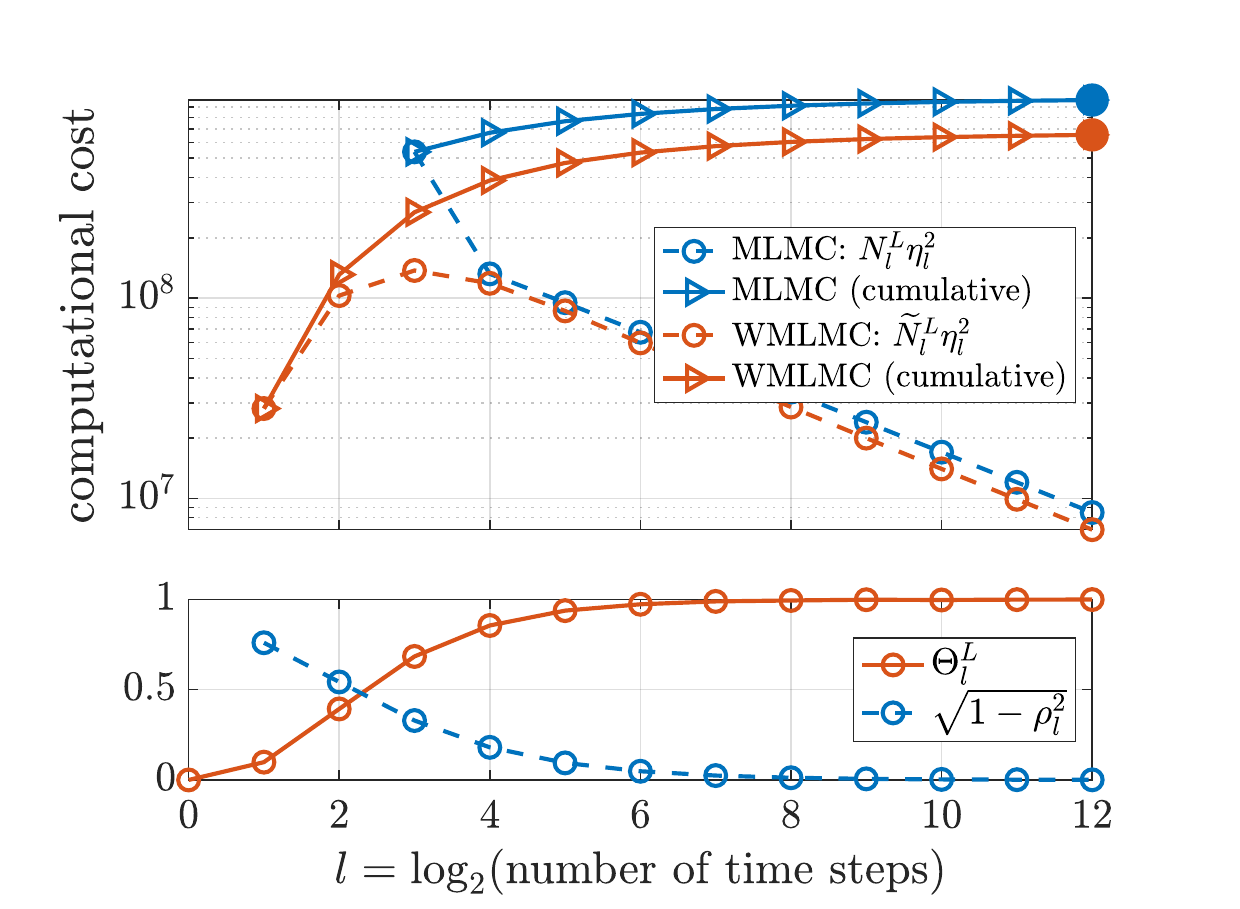}\\
\centerline{(b)}
  \end{minipage}
}
    \caption{Asian call option under GBM with Milstein discretisation, $M=2$. (a) Normalised costs for MLMC, WMLMC and single-level Monte Carlo (top) and the ratio between WMLMC and MLMC costs (bottom), shown as functions of the MSE $\epsilon^2$. (b) Contributions at  levels $l=0,\dots,12$ to the cost of WMLMC and MLMC estimates at level $L=12$ (top), with corresponding values of $\Theta^L_l$ and $\sqrt{1-\rho_l^2}$ (bottom). Note that because the correlatons on the coarsest grids are low, the optimal coarsest level is $l=3$ for MLMC and $l=1$ for WMLMC (so that $\Theta^{12}_0=0$). The ratio between the costs at level 12 is 1.49.}
\label{fig:AsianMilsteinrGBM2}
\end{figure}

In Figure~\ref{fig:AsianMilsteinrGBM2} we show results for the evaluation of an Asian call option under GBM, with the Milstein discretisation. Again $M=2$. This time, however, the correlations between the levels take a while to approach 1. This means, in particular, that for MLMC, it is optimal to take $l=3$ to be the coarsest level (forcing MLMC to use all the levels would result in higher overall cost). In contrast, WMLMC is able to take advantage of coarser-level estimates, with $l=1$ being the optimal choice of coarsest level. The effect of this can be seen in the left-hand graphs, which show that, by the time MLMC starts to differ from MC (at level $l=4$, with a MSE of just below $10^{-4}$), WMLMC has already started to make some efficiency gains. These then persist as the MSE decreases, leading to an eventual WMLMC cost of about 2/3 the MLMC cost.


\begin{figure}[htb!]
    \centerline{
\begin{minipage}{0.6\linewidth}
    \includegraphics[width=\linewidth]{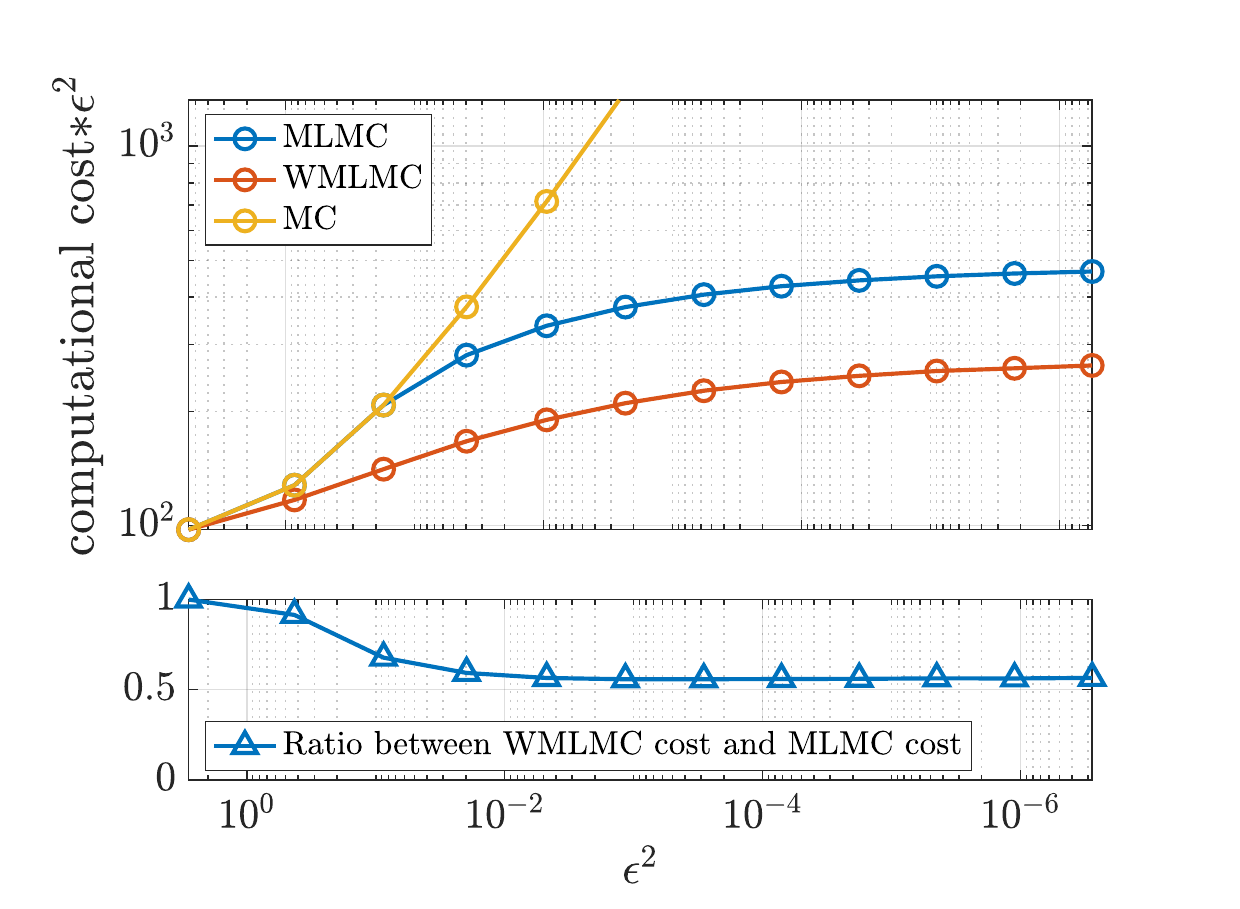}\\
\centerline{(a)}
  \end{minipage}
  \hspace*{-1cm}
\begin{minipage}{0.6\linewidth}
    \includegraphics[width=\linewidth]{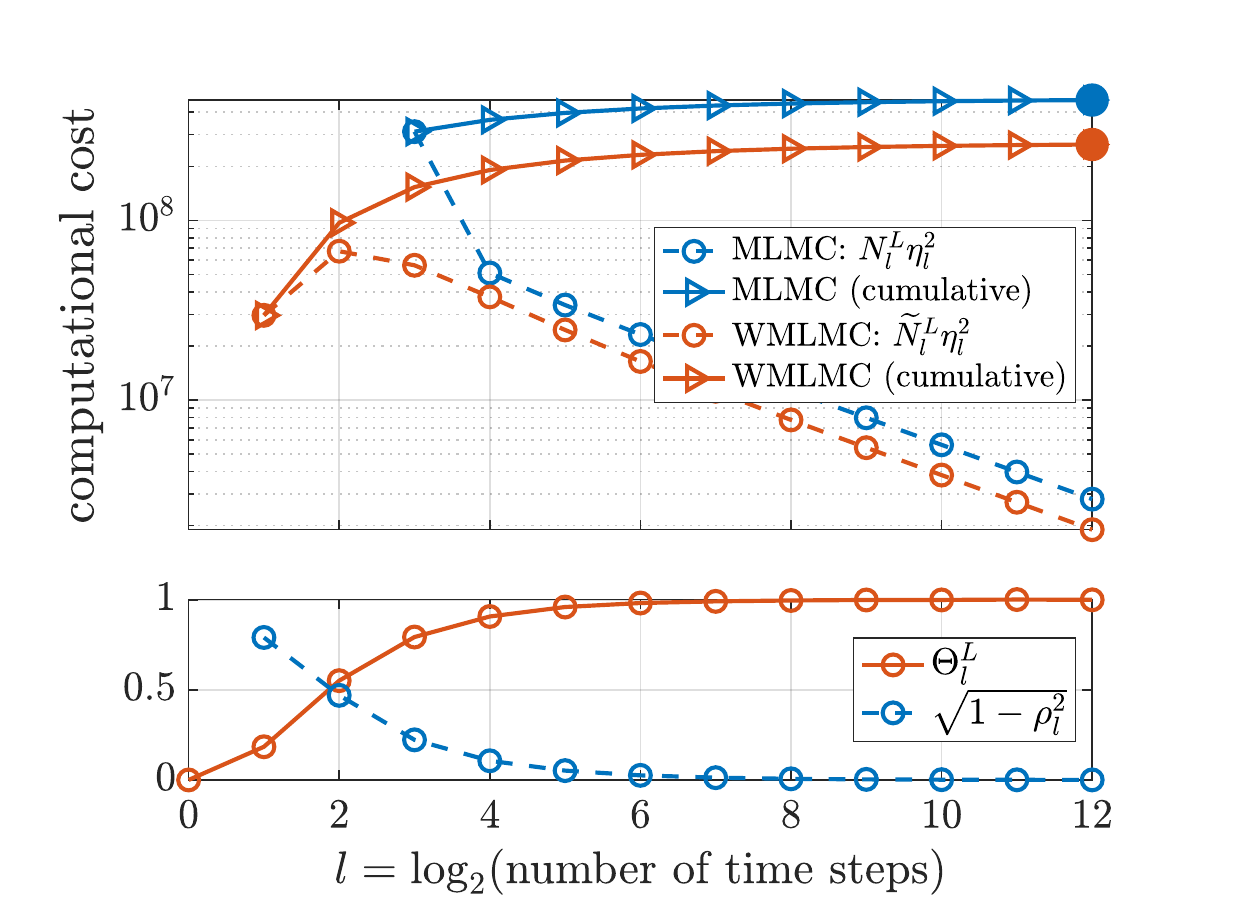}\\
\centerline{(b)}
  \end{minipage}
}
    \caption{European call option under IGBM with Milstein discretisation, $M=2$. (a) Normalised costs for MLMC, WMLMC and single-level Monte Carlo (top) and the ratio between WMLMC and MLMC costs (bottom), shown as functions of the MSE $\epsilon^2$. Note that in this case $\beta=2$ and $\gamma=1$ and so the asymtotic complexity is $O(\epsilon^2)$, and this is reflected in the (slow) levelling off of the costs. (b) Contributions at  levels $l=0,\dots,12$ to the cost of WMLMC and MLMC estimates at level $L=12$ (top), with corresponding values of $\Theta^L_l$ and $\sqrt{1-\rho_l^2}$ (bottom). Note that because the correlatons on the coarsest grids are low, the optimal coarsest level is $l=3$ for MLMC and $l=1$ for WMLMC. The ratio between the costs at level 12 is 1.77.} 
\label{fig:CallMilsteinIGBM2}
\end{figure}

Figure~\ref{fig:CallMilsteinIGBM2} shows results for a call option under IGBM, with a Milstein discretisation, and again $M=2$. Here the complexity is $O(\epsilon^2)$ for both MLMC and WMLMC. However, as in the previous example, WMLMC is able to make efficient use of coarser estimates than is possible for MLMC, resulting in the cost of WMLMC being less than $57\%$ of the MLMC cost.

\begin{figure}[htb!]
    \centerline{
\begin{minipage}{0.6\linewidth}
    \includegraphics[width=\linewidth]{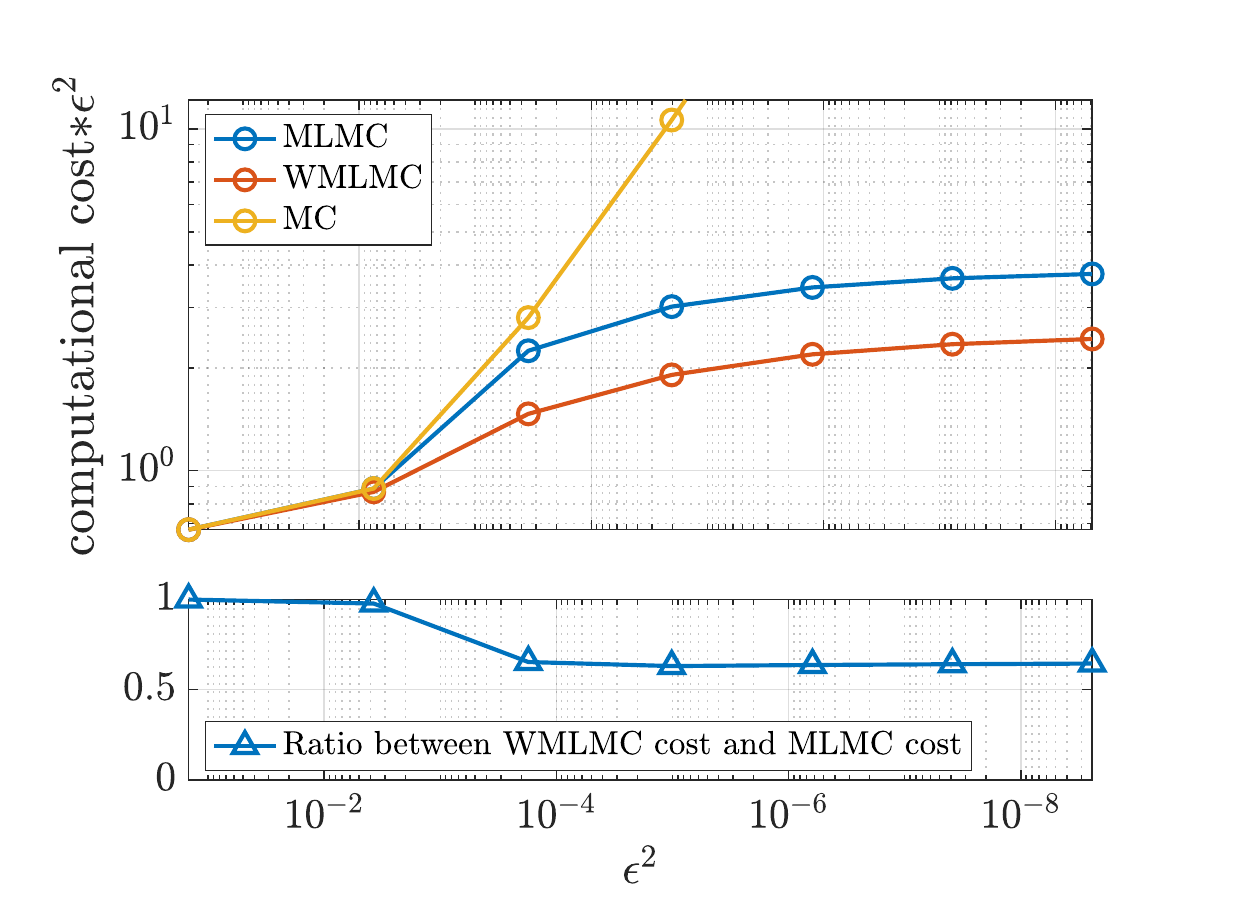}\\
\centerline{(a)}
  \end{minipage}
  \hspace*{-1cm}
\begin{minipage}{0.6\linewidth}
    \includegraphics[width=\linewidth]{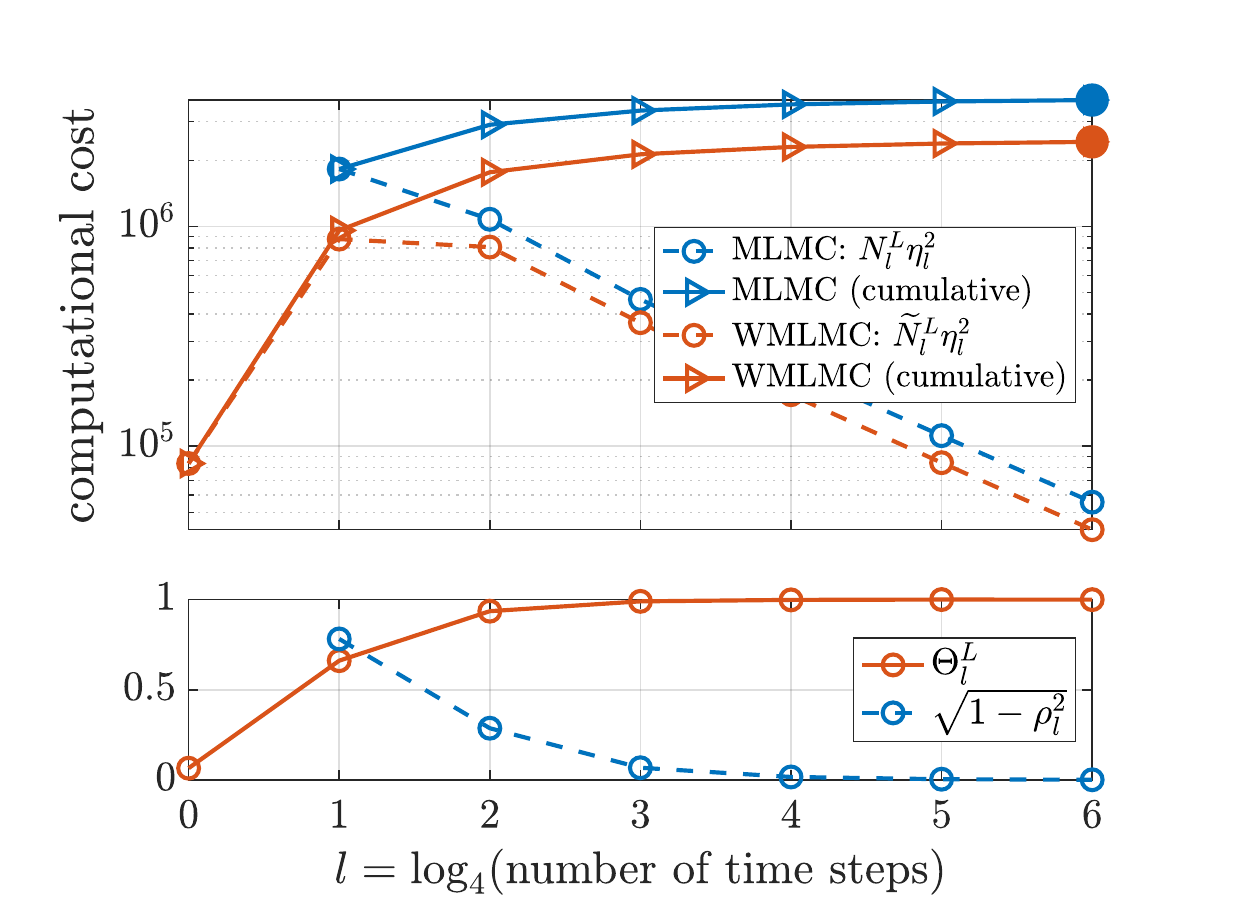}\\
\centerline{(b)}
  \end{minipage}
}
    \caption{European call option under CIR with Milstein discretisation, $M=4$. (a) Normalised costs for MLMC, WMLMC and single-level Monte Carlo (top) and the ratio between WMLMC and MLMC costs (bottom), shown as functions of the MSE $\epsilon^2$. Note that in this case $\beta=4$ and $\gamma=2$ and so the asymtotic complexity is $O(\epsilon^2)$. (b) Contributions at  levels $l=0,\dots,6$ to the cost of WMLMC and MLMC estimates at level $L=6$ (top), with corresponding values of $\Theta^L_l$ and $\sqrt{1-\rho_l^2}$ (bottom). The ratio between the costs at level 6 is 1.55.} 
\label{fig:CallMilsteinCIR2}
\end{figure}

The next experiment involves the CIR process with a Milstein discretisation, and the results are shown in Figure~\ref{fig:CallMilsteinCIR2}. The payoff is a European call, and we take $M=4$. In this case, $\beta=4$ and $\gamma=2$ and so the asymtotic complexity is $O(\epsilon^2)$. As in the previous experiments where this is the case, it shows up in the exponential decrease in level-wise costs in the top right figure, and in the levelling-off of the MLMC and WMLMC costs in the top left figure. Here both MLMC and WMLMC start to make gains over single-level MC estimates at $l=2$. However, the correlation is still relatively far away from 1 at this point, and the choice of optimal weights in WMLMC generates much more significant gains, so that the overall cost ends up being just under $65\%$ of the MLMC cost.

\begin{figure}[htb!]
    \centerline{
\begin{minipage}{0.6\linewidth}
    \includegraphics[width=\linewidth]{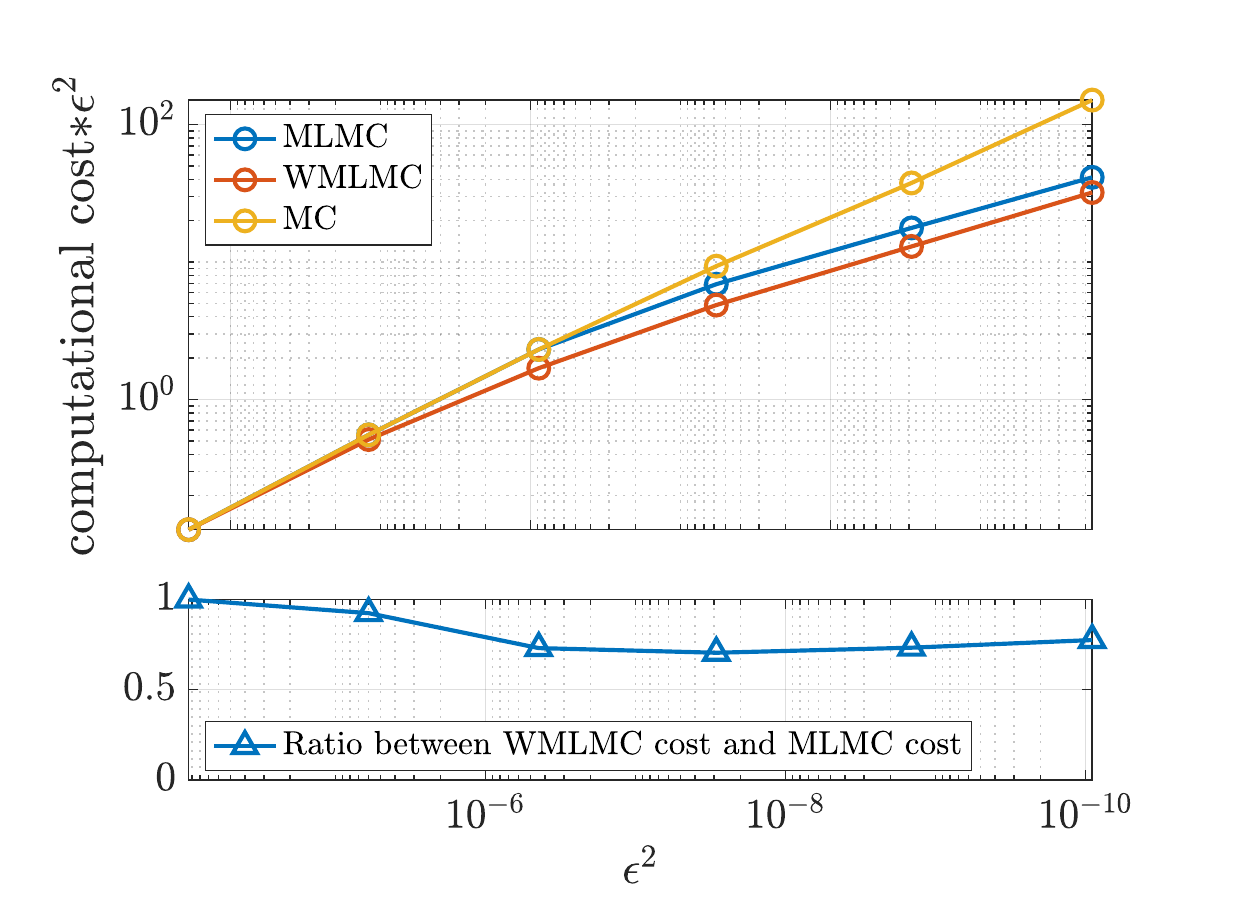}\\
\centerline{(a)}
  \end{minipage}
  \hspace*{-1cm}
\begin{minipage}{0.6\linewidth}
    \includegraphics[width=\linewidth]{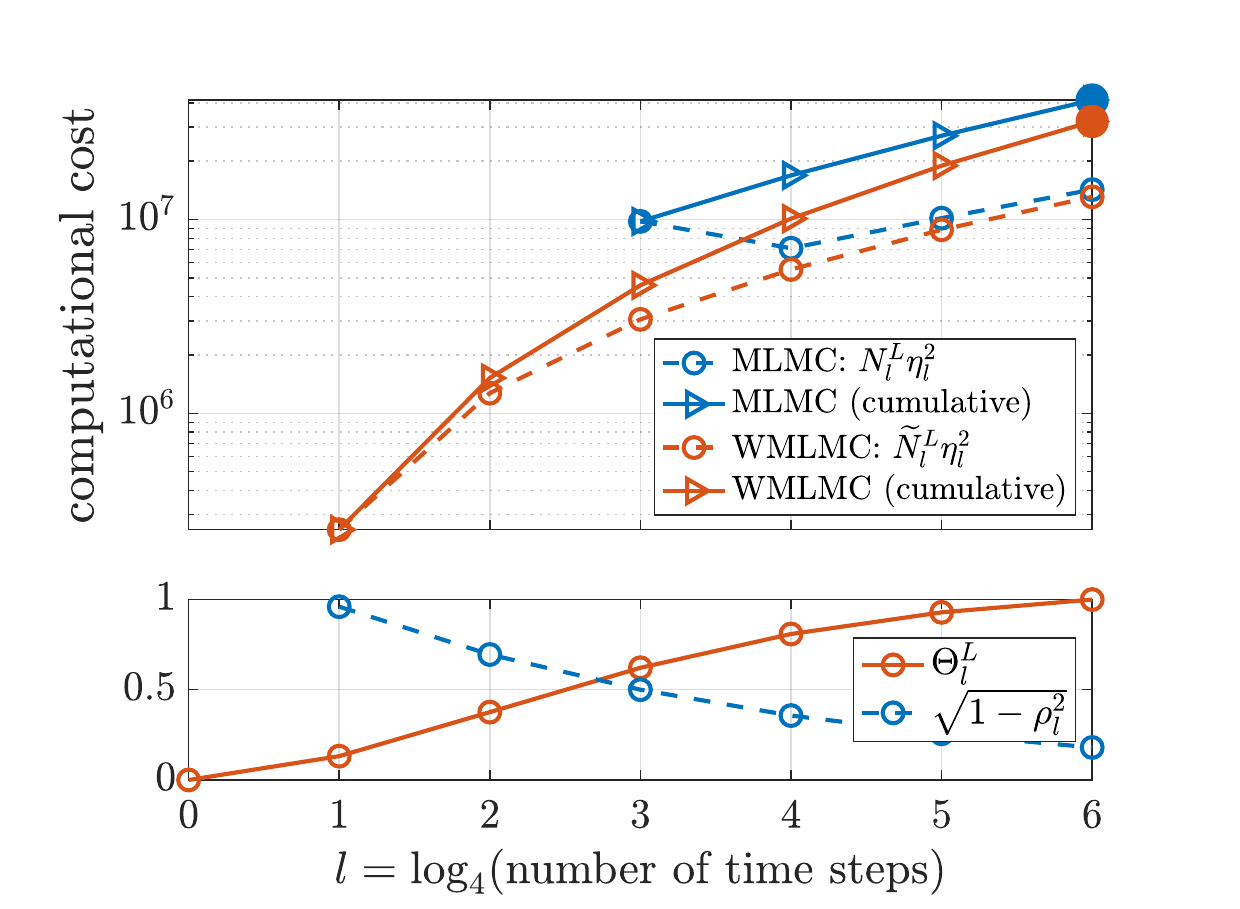}\\
\centerline{(b)}
  \end{minipage}
}
    \caption{Digital option under GBM with Euler-Maruyama discretisation, $M=4$. (a) Normalised costs for MLMC, WMLMC and single-level Monte Carlo (top) and the ratio between WMLMC and MLMC costs (bottom), shown as functions of the MSE $\epsilon^2$. Note that in this case $\alpha=2$, $\beta=1$ and $\gamma=2$ and so the asymtotic complexity is $O(\epsilon^{2.5})$. (b) Contributions at  levels $l=0,\dots,6$ to the cost of WMLMC and MLMC estimates at level $L=6$ (top), with corresponding values of $\Theta^L_l$ and $\sqrt{1-\rho_l^2}$ (bottom). The ratio between the costs at level 6 is 1.29} 
\label{fig:DigitalEulerGBM4}
\end{figure}

In Figure~\ref{fig:DigitalEulerGBM4}, we illustrate a situation where the gains from MLMC (and WMLMC) are relatively small. This can be seen to occur with a naive implementation of a digital option payoff (as a result of the discontinuity in the payoff). We show results for an Euler Maruyama discretisation of GBM, with $M=4$. Here again WMLMC is able to take \emph{some} advantage of coarse estimates with low correlations, and the cost is around $75\%$ of the MLMC cost.

In order to illustrate the potential performance of the MLML and WMLMC methods, 
the experiements so far have been run under `ideal' conditions, using $10^6$ samples at each level, so that we have been able to compute the variances, correlations and associated values of $\theta_l$ and $\Delta_l$, etc.~with very little uncertainty. As mentioned at the beginning of this section, in practice an algorithm along the lines of Algorithm~1 of \cite{Giles2015}. For our final experiment, we implemented such an algorithm 1000 times on the IGBM call option valuation with Milstein discretisation and antithetic variates, using an initial 20 samples at each level (to generate initial estimates of the bias, covariance and correlation). The target MSE was $10^{-6}$, and we aimed to split the error evenly between the contributions from the bias and from the variance. In Figure~\ref{fig:Histograms} we show histograms of the resulting MLMC and WMLMC estimates (left) along with histograms of the computational effort required (i.e.~$\sum_{l=0}^LN^L_l\eta^2_l$; the value of $L$ determined by the algorithm was either $11$ or $12$, depending on the run). The average costs were $4.785\times 10^8$ for MLMC, and $2.831\times 10^8$ for WMLMC, a ratio of about 1.69 (only slightly less than the value of 1.77 obtained under idealised conditions, as shown in Figure~\ref{fig:CallMilsteinIGBM2}). The MSE values, estimated from \eqref{eq:MSE} using a value for $\overline{P}$ computed using WMLMC with a target MSE of $2.5\times 10^{-9}$, were $1.11\times 10^{-6}$ for MLMC and $1.06\times 10^{-6}$ for WMLMC.

\begin{figure}[htb!]
    \centerline{
\begin{minipage}{0.6\linewidth}
    \includegraphics[width=\linewidth]{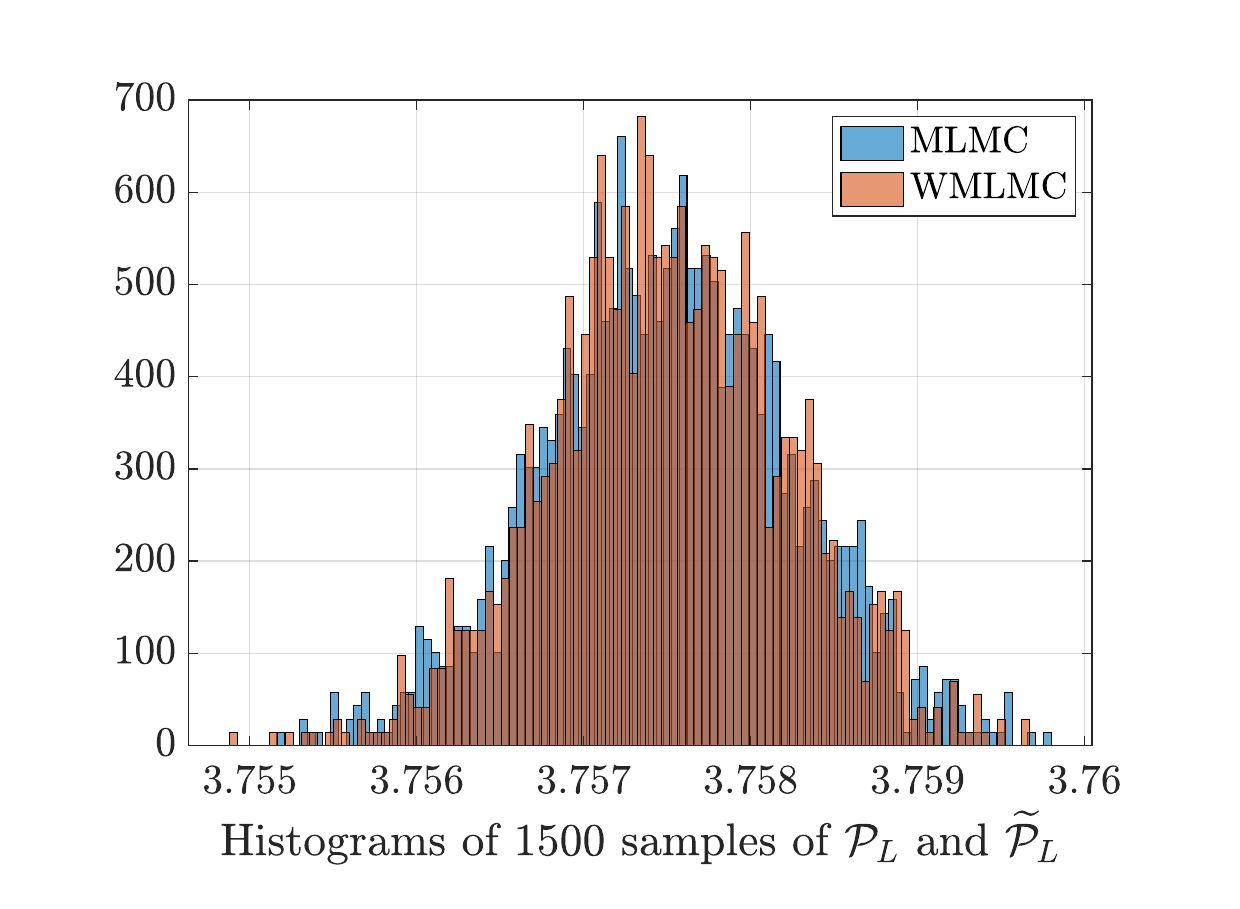}\\
\centerline{(a)}
  \end{minipage}
  \hspace*{-1cm}
\begin{minipage}{0.6\linewidth}
    \includegraphics[width=\linewidth]{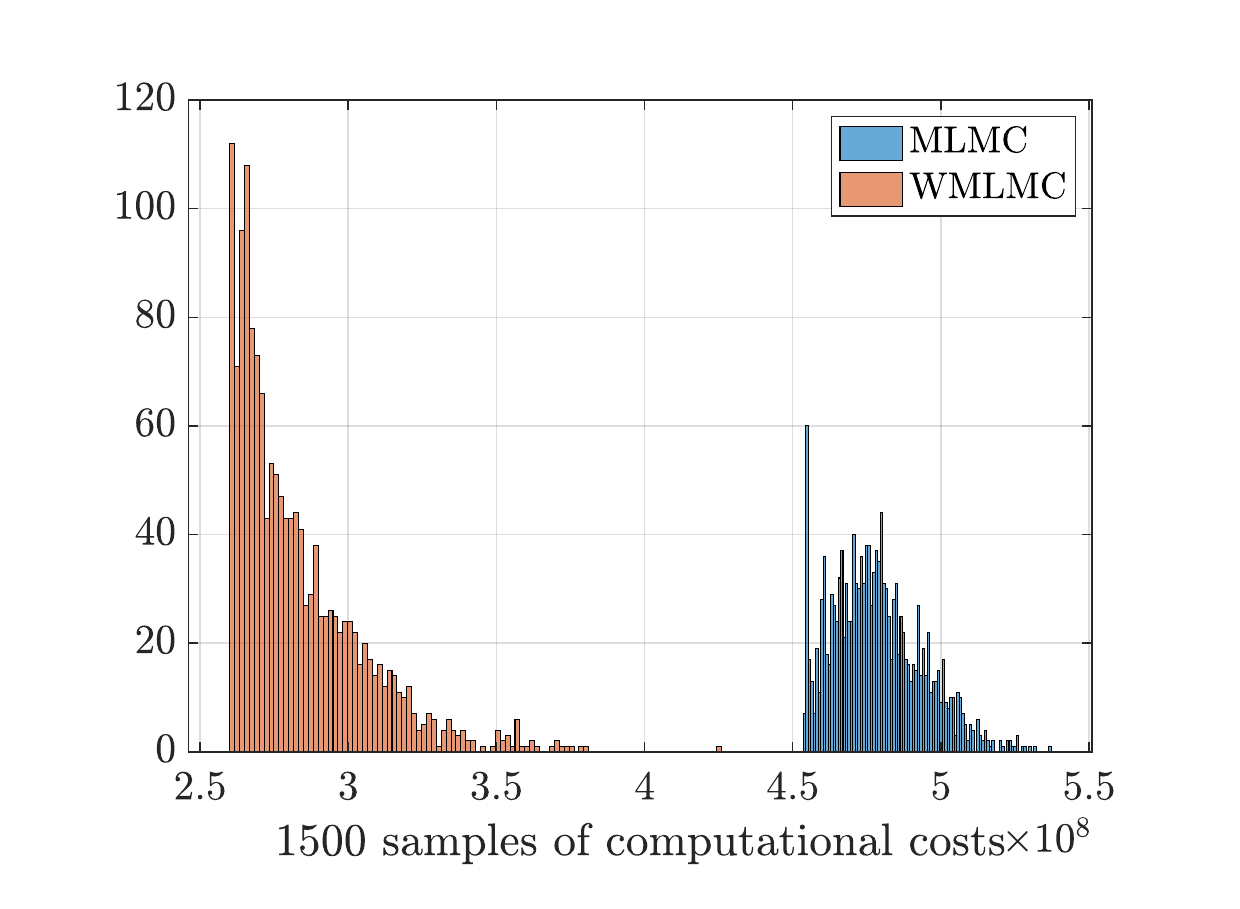}\\
\centerline{(b)}
  \end{minipage}
}
    \caption{Results from the computation of 1500 MLMC and WMLMC estimates of a call option under IGBM with a Milstein discretisation, with $M=2$, and a target MSE of $10^{-6}$. The left hand graph (a) shows histograms of the two sets of estimates. The estimated MSE values are $1.11\times 10^{-6}$ for MLMC, and $1.06\times 10^{-6}$ for WMLMC. The right hand graph (b) shows histograms of the computational costs $\sum_{l=0}^LN^L_l\eta^2_l$ for the two estimators. The value of $L$ determined by the algorithm was either $11$ or $12$, depending on the run.} 
\label{fig:Histograms}
\end{figure}

\section{Discussion}
The addition of weights to the standard MLMC (or MIMC) approach is based on viewing the multilevel estimator as a nested sequence, with each estimator in the sequence using a lower-level multilevel estimator as a control variate. The computation of the optimal weights creates some additional overhead cost, but this is marginal, and so the efficiency savings essentially come `for free'.

As mentioned in the introduction to Section~\ref{sect:NumExpts}, a weighted MLMC approach was proposed in \cite{LemairePages2017}. The weights there are used to improve the bias, rather than the variance, and can be seen as complementary to the weights used here. Indeed, the two approaches can be combined, and the potential gains from this combination are explored in \cite{Li2024}. 

The extension to the multi-index case has been outlined here in the case where the set of indices is in tensor product form. Initial experiments in this setting \cite{WareLi2024} indicate that the potential gains in efficiency are even higher than in the single-index case. Future work will explore this more fully (see also \cite{Li2024}), along with extensions to a weighted version of the non tensor-product MIMC implementation.

 \section*{Acknowledgments}
We acknowledge the support of the Natural Sciences and Engineering Research Council of Canada (NSERC).\\ 
Nous remercions le Conseil de recherches en sciences naturelles et en génie du Canada (CRSNG) de son soutien.

\bibliography{References.bib}
\bibliographystyle{abbrv}

\end{document}